\documentclass[11pt]{article}

\usepackage{libertine}
\usepackage{fullpage}

\usepackage{amsmath}
\usepackage{amsfonts}
\usepackage{amssymb}
\usepackage{amsthm}
\usepackage{amsbsy}

\usepackage{graphicx}
\usepackage{subcaption}
\usepackage{verbatim}
\usepackage{url}
\usepackage{bbm}
\usepackage{multirow}
\usepackage{bm}
\usepackage{xcolor}

\usepackage[ruled,longend]{algorithm2e} 

\newtheorem{thm}{Theorem}
\newtheorem{prop}{Proposition}

\theoremstyle{definition}

\newtheorem{example}{Example}

\newcommand{\calS}{\mathcal{S}}
\newcommand{\calL}{\mathcal{L}}
\newcommand{\calI}{\mathcal{I}}
\newcommand{\calR}{\mathcal{R}}

\newcommand{\calD}{\mathcal{D}}

\newcommand{\calU}{\mathcal{U}}
\newcommand{\calQ}{\mathcal{Q}}
\newcommand{\calH}{\mathcal{H}}

\newcommand{\calW}{\mathcal{W}}

\newcommand{\haschild}{\rightsquigarrow}
\newcommand{\total}{\text{total}}
\newcommand{\one}{\mathbf{1}}
\DeclareMathOperator{\sgn}{sgn}

\newcommand{\remove}[1]{}

\usepackage{authblk}
\usepackage{natbib}

\begin{document}

\title{\bf Energy-aware tree network formation \\ among computationally weak nodes\thanks{
A preliminary version of this paper entitled ``Peer-to-peer energy-aware tree network formation'' appeared in {\em Proceedings of the 16th ACM International Symposium on Mobility Management and Wireless Access (MOBIWAC)}, pages 1--8, 2018~\citep{MNTV18}. 
This full version extends the conference one in multiple ways. It contains all missing proofs of the theoretical statements and additional details on the correctness and efficiency of our methods for the lossless case. 
Besides presenting protocols only for the formation of arbitrary and binary tree networks, we now further present protocols for the formation of $k$-ary tree networks, for any integer $k\geq 2$, building upon ideas used for binary trees. For such networks, we also show that it is possible to achieve particular energy distributions when there is no energy loss, by designing new protocols for these cases as well. 
 The evaluation of our methods in the lossy case has been extended to compare the efficiency of the protocols in terms of the energy that has been lost until a stable energy distribution has been reached.  
Our discussion on related work, generalizations of our model and open problems has also been extended significantly. Moreover, many examples have been included throughout the paper to illustrate how all proposed protocols work in particular cases. 
We would like to thank Dimitrios Tsolovos for fruitful discussions at early stages of this work that led to the publication of the preliminary conference version of this paper. 
This work has been partially supported by the Greek State Scholarships Foundation (IKY), and by the Alexander S. Onassis Public Benefit Foundation.}}

\author{Adelina Madhja \quad Sotiris Nikoletseas}
\affil{\em Department of Computer Engineering and Informatics, University of Patras, Greece \\
\em Computer Technology Institute and Press ``Diophantus" (CTI), Greece}

\author{Alexandros A. Voudouris}
\affil{\em Department of Computer Science, University of Oxford, UK}

\date{}

\maketitle 

\begin{abstract}
We study the fundamental problem of distributed network formation among mobile agents of limited computational power that aim to achieve energy balance by wirelessly transmitting and receiving energy in a peer-to-peer manner. Specifically, we design simple distributed protocols consisting of a small number of states and interaction rules for the formation of arbitrary and $k$-ary tree networks. Furthermore, we evaluate (theoretically and also using computer simulations) a plethora of energy redistribution protocols that exploit different levels of knowledge in order to achieve desired energy distributions among the agents which require that every agent has exactly or at least twice the energy of the agents of higher depth, according to the structure of the network. Our study shows that without using any knowledge about the network structure, such energy distributions cannot be achieved in a timely manner, meaning that there might be high energy loss during the redistribution process. On the other hand, only a few extra bits of information seem to be enough to guarantee quick convergence to energy distributions that satisfy particular properties, yielding low energy loss.  
\end{abstract}

\section{Introduction}
During the last few decades, the problem of how to carefully exploit the technology of Wireless Power Transfer (WPT)~\citep{cannon09,magnetic,battery,IH11,WCS04} in (mobile) ad hoc networks has been one of the main subjects of study in the related literature on sensor networks. In particular, special nodes can be deployed in the network area and exploit WPT in order to replenish the energy reserves of the network nodes (which drain their batteries while performing sensing and communication tasks), allowing them to sustain their normal operation for longer periods of time, and thus considerably extending the network lifetime; see~\citep{nikole} for an introduction to this topic. 

Further recent developments on WPT-related technologies offer mobile devices the capability to achieve bi-directional energy transfer, and enable peer-to-peer energy exchanges between network nodes \citep{dPCGCMP15, SCP15}, thus motivating new applications, like energy sharing between electric vehicles and portable devices.
Towards this goal, \citet{WKFP16} recently designed PowerShake, a WPT-based power sharing system between mobile devices, that allows the users to have control over the energy in their personal devices and {\em trade it} with others on demand. Such a system can be applied in vehicular ad hoc networks (VANETS)~\citep{vanet,vanet2,vehicle1,vehicle2}, especially with the very recent development and massive demand of electric vehicles.

Of course, one can think of many other possible and more critical applications of this technology. For instance, imagine a number of tiny medical devices, equipped with the appropriate sensors, which are injected into a patient's body in order to monitor her medical status. Due to their size, such devices would have to be computationally weak, with limited memory and energy reserves. However, by carefully exploiting their communication capabilities, these devices can coordinate with each other in order to form complex network structures with the goal of performing more advanced computations and also help each other survive by exchanging energy. 

\citet{MS14} initiated the study of network formation among populations of computationally weak agents. Their model is inspired by the population protocol model of \citet{AAE07} and the mediated population protocol of \citet{MCS11}, and assumes that the agents do not share memory or exchange messages unless they interact, in which case they can connect to each other and form particular network structures. The agents are required to collectively {\em converge} to a stable state, even though they cannot grasp the status of the entire population. The authors design, prove the correctness, and analyze the time complexity of a series of protocols for the formation of many interesting networks focusing, among others, to stars, lines and rings. Moreover, they design generic protocols that are capable of simulating Turing Machines in order to construct large classes of networks.   

Achieving energy balance over the network nodes using peer-to-peer WPT methods can ensure the stable operation of the network and the extension of its lifetime. Recently, \citet{NRR17} showed that energy balance can be achieved by populations of devices (agents) that interact in an opportunistic manner. Whenever two agents happen to come in close proximity they can decide whether to exchange energy based on their current configurations (energy levels, memory, etc). The authors simulate the interactions between pairs of nodes by assuming the existence of a probabilistic scheduler whose purpose is to decide which pairs of nodes will interact in every step of time. Moreover, they assume that during an energy exchange, a {\em constant} fraction of the transmitted energy is lost.

\citet{madhja_mswim2016} studied the problem of energy-aware network formation among populations of computationally weak agents, a problem that combines elements from the models of \citet{MS14} and \citet{NRR17}. The goal is to design distributed protocols so that the nodes connect with each other in order to form a {\em star} network and achieve a particular energy distribution, where the central node stores half of the network energy, while the remaining energy is evenly distributed among the rest of the nodes. Another critical difference of the model of \citet{madhja_mswim2016} from that of \citet{NRR17} is that the energy that is lost during any exchange {\em varies} from interaction to interaction; specifically, in their simulations, \citet{madhja_mswim2016} treat the loss of energy as a random variable following a normal probability distribution with given expected value and variance. 

\subsection{Our contribution}
In this paper, we adopt the model of \citet{madhja_mswim2016} and consider scenarios involving a population of computationally weak mobile agents aiming to (1) distributively form a {\em tree} network structure, and (2) achieve some sort of energy balance. In particular, we focus on the construction of arbitrary and $k$-ary tree networks (for integer $k \geq 2$), and energy distributions with the property that every node has exactly or at least twice the energy of each of its children. 

We start by designing, proving the correctness, and analyzing the time complexity of two simple protocols for the construction of arbitrary and binary tree networks; these protocols require only a few states (four for arbitrary trees and six for binary) and interaction rules (five and eight, respectively). 
Then, we extend the second protocol (used for the construction of binary trees), so that the agents can form more general $k$-ary tree networks.
We also present two simple interaction rules that allow the nodes to locally infer useful information about their own depth (distance from the root node) and the height of the tree. 

To achieve the desired energy distributions, which depend on the network structure, we present of series of energy redistribution protocols, which differ on the amount of energy that is exchanged among the agents during their interactions, and on the knowledge that the agents need to have about the network structure in order to achieve the distributions. Our goal is to test the limits of what type of energy distributions are possible to achieve using devices with different computational capabilities. 
 
\begin{itemize}
\item 
For the {\em lossless} case, where no energy is lost during the interactions among the agents, we theoretically showcase the properties of our protocols; this assumption allows us to argue whether, even under unrealistically perfect conditions, our protocols are able to converge to the desired energy distributions. In particular, we show that one of our protocols is able to achieve an energy distribution according to which each node has {\em exactly} twice the energy of its children. However, since this protocol requires complete information about the structure of the network, we then focus on protocols that are oblivious to such information or utilize some of it, and show that they can achieve relaxed energy distributions (where each node has at least twice the energy of its children, with the possible exception of the root node of the tree) in some cases.

\item
For the {\em lossy} case, where energy is lost during the interactions between the nodes, we conduct simulations in order to fine-tune several parameters used by our protocols, and evaluate their performance on several metrics, which aim to measure the convergence time of the protocols and the quality of the outcome energy distribution compared to the ideal one. Surprisingly, we observe that two of our oblivious protocols (that do not require any global information about the network) match, and sometimes outperform, our stronger protocols, which require some more concrete knowledge about the network. 
\end{itemize}

\subsection{Other related Work}
Wireless Power Transfer has been extensively studied in the context of (mobile) ad hoc networks. In most of these studies, powerful chargers are used with the sole purpose of replenishing the energy of the network nodes.
For instance, \citet{MNR15} consider multiple mobile chargers with limited energy and design efficient traversal and coordination strategies with the goal of extending the network lifetime of static sensor networks. In contrast, \citet{ABER15} consider mobile ad hoc networks and a single mobile charger with infinite energy that traverses the network in order to recharge the agents as required. Another variation of this flavor, is the recent work of \citet{MNV18} who consider mobile ad hoc networks and a single static charger that has the ability to adapt its charging power in order to balance the trade-offs between recharging the mobile agents and saving energy for future interactions.
\citet{nikoletseas2017wireless} experimentally show that WPT is a viable option, by performing experiments using real devices that aim to charge the sensors and keep the energy loss low. They also propose a protocol that achieves energy balance over the chargers. 
For a more extensive overview of WPT techniques as well as some of the key elements that make WPT possible, we refer the interested reader to the work of \citet{bi2016wireless} and the book by \citet{nikole}.

\citet{collaborative}, \citet{comnet_madhja2016} and \citet{collaborative2} considered collaborative WPT schemas and assumed some aspects of peer-to-peer energy exchange between the chargers. Specifically, the authors proposed protocols that allow the chargers, besides charging the networks nodes as usual, to also cooperate and charge each other. 
\citet{BS17} studied energy sharing in mobile social networks, by taking into account the charging patterns of devices and the social interactions between their owners. They also proposed an energy sharing model according to which the nodes of the mobile networks are paired into power buddies.
\citet{dhungana2018charging} studied a model according to which the mobile devices can be charged either by using a charging cable or in a peer-to-peer manner based on their interactions with other devices and the goal is to minimize the number of times a device needs to be charged via the cable. 
\citet{bulut2018crowdcharging} studied the potential of crowdcharging. They discussed its feasibility, the software and hardware challenges that emerge for its use, and also developed an application that builds a social network among the users and manages the entire process of power sharing between the mobile devices. 

\subsection{Roadmap}
The rest of the paper is structured as follows.
In Section~\ref{sec:prelim} we give necessary preliminary definitions and examples involving the tree network structures and the energy distributions that we aim to achieve. In Section~\ref{sec:tree-protocols} we present and analyze various tree network formation protocols as well as protocols that allow the agents to locally learn useful network characteristics. In Section~\ref{sec:energy-protocols} we present simple redistribution protocols and theoretically study their properties when there is no energy loss. Then, in Section~\ref{sec:experiments}, using computer simulations, we study the properties of our protocols even when there is random energy loss. We conclude in Section~\ref{sec:conclusion} with a short discussion on possible extensions of our model and many interesting open problems. 

\section{Preliminaries}\label{sec:prelim}
We consider a population of $n$ mobile agents (or {\em nodes}) $V=\{v_1, v_2, ..., v_n\}$ of limited computational power and memory, who move around in a bounded networking area. When two nodes come in close proximity and enter the communication range of each other, they interact according to an {\em interaction protocol} which essentially defines the information that the nodes must exchange, and how they should update their configurations. The objective of the interaction protocol is the nodes to eventually distributively form a desired network structure and achieve an energy distribution. In the following, we give necessary definitions and terminology regarding node configuration, interactions, tree formation, energy distributions, exchanges and loss. We also present a few examples to help the reader fully grasp the details of the setting.  
 
\subsection{Configurations and interactions} 
We assume that the time runs in discrete steps. For every time step $t \in \mathbb{N}_{\geq 0}$, each node $v$ is in a state $q_v(t)$ from a set $\calQ$ of possible states, has energy $E_v(t)$, memory $M_v(t)$, and network connections $N_v(t)$; the tuple $c_v(t) = [q_v(t),E_v(t),M_v(t),N_v(t)]$ is the {\em configuration} of node $v$ at time $t$. During each time step $t$, a pair of nodes $(u,v)$ interacts, and the nodes update their configurations according to the rules of the interaction protocol, which define {\em transitions} of the form
$$(c_u(t),c_v(t),\text{\em condition}) \rightarrow (c_u(t+1),c_v(t+1)).$$
Such a transition indicates that, if a specific condition regarding the nodes or the network is met, then the configurations of the two participating nodes $u$ and $v$ are updated from $c_u(t)$ and $c_v(t)$ at time $t$ to $c_u(t+1)$ and $c_v(t+1)$ at time $t+1$, respectively. 

In general, the movement of the nodes can be highly arbitrary as they may correspond to smart devices that are carried around by humans or robots following possibly unpredictable movement patterns performing various tasks. Following previous work~\citep{madhja_mswim2016,MS14,NRR17}, we abstract the movement of the nodes by assuming that all interactions are planned by a {\em fair scheduler}, which satisfies the property that all possible interactions will eventually occur. Specifically, we consider the existence of a {\em fair probabilistic scheduler}~\citep{AAE07} according to which, during every time step, a single pair of nodes is selected independently and uniformly at random among all possible pairs of nodes in the population, and the selected nodes then interact following the rules of the interaction protocol. Observe that, due to the behavior of the scheduler, every pair of nodes will most probably interact in a span of $\Theta(n^2)$ time steps.

We should emphasize here that the purpose of the scheduler is to {\em indicate} which nodes are close to each other at each time step and can therefore interact. Essentially, the scheduler hides the movement and location characteristics of the nodes and allows us to avoid defining any particular mobility model to specify how the location, speed, and direction of the nodes changes over time.

\subsection{Tree formation}
Our goal in this paper is to define interaction protocols that can construct tree networks. Hence, we assume that the nodes can form parent-child connections when they interact with each other and are allowed to do so according to the rules of the protocol, depending on their configurations. 
When such a connection is established between two nodes $u$ and $v$, a {\em directed} edge from $u$ to $v$, denoted as $u \haschild v$, is formed in the network. 
The direction indicates that $u$ is the parent node and $v$ is the child node, and helps both nodes realize who of them is the parent when they communicate again in future interactions. 
Each node can have at most one parent, but multiple children, depending on the tree we are aiming for. 

We are interested in rooted spanning tree structures consisting of {\em all} nodes in the network; note that the exact number of the nodes in the network is considered unknown and our network formation protocols do not rely on such global information.  
Our goal is to construct {\em arbitrary} trees where each node may have any number of children, and {\em $k$-ary} trees (for any integer $k \geq 2$) where every node can have at most $k$ children. 
We call a node $v$ {\em isolated} if it is not connected to any other node in the network, {\em leaf} if it has a parent but no children, {\em internal} if it has a parent and at least one child, and {\em root} if it does not have a parent but has children.

\subsection{Energy distributions}\label{sec:def-distributions}
Another goal of this paper is to define protocols that can achieve desired energy distributions over the nodes of the network. We define three interesting types of energy distributions. All of them demand that nodes of lower depth (closer to the root of the tree) have more energy than nodes of higher depth. This is intuitive in applications related to {\em data propagation} and {\em energy-efficient routing}, where the tree is used as the communication network according to which the nodes can only communicate with the nodes they are connected to (their parents and children; the only nodes they aware of). It is intuitive that nodes that are closer to the data sink (according to the network) are involved in almost all propagations and need plenty of energy to function properly for long periods of time, while distant nodes rarely communicate that much and they can sustain normal operation using less energy~\citep{survey1,survey2,survey3,survey4,routing1}. 
\footnote{We remark that most energy-efficient routing protocols in wireless sensor networks aim to improve metrics such as the network lifetime. While the particular energy distributions we define here can also be thought of as aiming to extend the network lifetime, our primary goal is to achieve energy balance among the nodes that is consistent with the type of network that is formed, which in this paper is a tree.}

The first energy distribution that we consider is the strongest one (in terms of the condition that must be satisfied) and requires that every node has {\em exactly} twice the energy of each of its children. We refer to such an energy distribution as {\em exact}. In particular, we say that a parent-child pair of nodes $(p,c)$ is in an {\em exact energy equilibrium} at time $t$ if $E_p(t) = 2 E_c(t)$. Then, a tree network {\em converges} to an exact energy distribution at time $t$ if every parent-child pair of nodes $(p,c)$ is in an exact energy equilibrium. 

The second type of energy distributions is a relaxation of the first one and requires every node to have {\em at least} twice the energy of each of its children. We refer to such distributions as {\em relaxed}; observe that there are infinitely many relaxed energy distributions, in contrast to the exact distribution which is unique. In particular, we say that a parent-child pair of nodes $(p,c)$ is in a {\em relaxed energy equilibrium} at time $t$ if $E_p(t) \geq 2 E_c(t)$. Then, a tree network converges to a relaxed energy distribution at time $t$ if every parent-child pair of nodes $(p,c)$ is in a relaxed energy equilibrium. 

Finally, the third type we consider is yet another relaxation of the exact energy distribution, which is however more restrictive than the relaxed one. According to this type of energy distributions, we require that every node has exactly twice the energy of each of its children {\em except} possibly for the root of the tree, which can have more or less than twice the energy of its children; essentially the root assists the other nodes of the network achieve an exact distribution. We refer to such an energy distribution as {\em exact up to the root}. 

\subsection{Energy exchange and loss}
To achieve the aforementioned energy distributions, the nodes must exchange energy when they meet and interact (according to the choices of the scheduler of course). A {\em redistribution protocol} defines the conditions under which such exchanges take place as well as the amount of energy that one node has to transfer to another during their interaction. Energy transfer in such a peer-to-peer manner is possible by carefully utilizing the WPT technology, which we here use as a black box.

Due to the nature of wireless energy technology, any energy exchange may induce some irrevocable {\em energy loss}. Following~\citep{madhja_mswim2016,NRR17}, we assume that whenever a node $v$ is supposed to transfer an amount $x$ of energy to another node $u$, a fraction $\beta$ of $x$ is lost, and $u$ actually receives $(1-\beta) x$ units of energy. In general, $\beta$ is unknown and may vary from interaction to interaction (during which there is indeed energy exchange among the participating nodes -- not all interactions lead to energy exchange). When we study the theoretical properties of our protocols, we mainly focus on the {\em lossless} case where $\beta=0$. However, in our simulations, we additionally consider the {\em lossy} case, where $\beta$ is a random variable following some normal distribution. 
We remark that we do not account for the possible energy loss due to movement or other activities of the nodes, as this can be supplied by other sources: the nodes may be carried around by humans, or robots which can use gas (or some other type of energy) in order to move. 

\subsection{An example}
Let us present a quick example to fully understand the tree network structures and the energy distributions we are aiming for. 

\begin{example}\label{example:trees-distributions}
Consider an instance with $n=7$ nodes such that they initially have a total of $100$ units of energy (for the purposes of the example, it is irrelevant how this total energy is split among the nodes), and assume that $\beta=0$ (no energy is lost when energy is exchanged). Figure~\ref{fig:example} depicts three possible tree networks and three different energy distribution that can be formed while the nodes interact with each other.
\begin{itemize}
\item The first network is an arbitrary tree consisting of one root node and $6$ leaf nodes; this special type of a tree is known as a star: the root is the center and the leaves are peripherals. The energy distribution is exact since the root node stores twice the energy of each of the other nodes. 
\item The second network is a binary tree where every node has at most two children, and the energy distribution is relaxed since every node has at least twice the energy of each of its children.
\item The third network is a $3$-ary tree where every node has at most three children, and the energy distribution is exact up to the root since only the root does not have twice the energy of each of its children.  
\end{itemize} 
Observe that the exact energy distribution is unique and all nodes at the same depth (distance from the root) end up having exactly the same energy (in our example this is equal to $12.5$). This is not true for relaxed and exact up to the root energy distributions as one can observe from our example, where nodes of the same depth might end up storing different amounts of energy (for instance, in the second network the root's children have $15$ and $21$ units of energy, respectively).
\hfill
\qed 
\end{example}

\begin{figure}[t!]
\centering
   \frame{\includegraphics[page=1,scale=1.1]{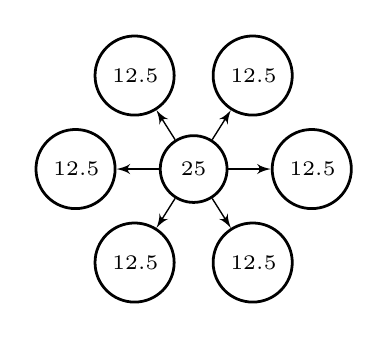}}
   \ \ \
   \frame{\includegraphics[page=2,scale=1.1]{example}}
   \ \ \
   \frame{\includegraphics[page=3,scale=1.1]{example}}
\caption{The tree networks and energy distributions of Example~\ref{example:trees-distributions}; the number in each node represents the units of energy that the node stores so that the sum over all nodes is equal to $100$.
(a) the left network is an arbitrary tree (star) and the nodes achieve an exact energy distribution (the root node has exactly twice the energy of every other node). 
(b) the middle network is a binary tree and the nodes achieve a relaxed energy distribution (all nodes have at least twice the energy of their children). 
(c) the right network is a $3$-ary tree and the nodes achieve an exact up to the root energy distribution (all nodes but the root have exactly twice the energy of their children).}
\label{fig:example}
\end{figure}

\section{Tree network formation}\label{sec:tree-protocols}
In this section, we present distributed protocols for the formation of tree networks. We start with arbitrary trees and binary trees, and then extend the ideas used for binary trees to the more general case of $k$-ary trees, for any integer $k \geq 2$. Finally, we present two interaction rules that can be used by the nodes in order to learn their own depth  and the total height of the tree. 

\subsection{Arbitrary trees}\label{sec:arbitrary-tree}
For the formation arbitrary trees, we consider the protocol {\bf TreeConstructor} which uses $4$ states and the interaction rules of Table~\ref{TreeConstructor-rules}; see also Figure~\ref{fig:tree_rules} for a graphical explanation of the rules of the protocol. In particular, we say that a node $v$ is in state 
\begin{itemize}
\item $S$ if $v$ is isolated,
\item $L$ if $v$ is a leaf,
\item $I$ if $v$ is internal, and
\item $R$ if $v$ is a root.
\end{itemize}
According to the interaction rules, when two nodes $u$ and $v$ interact with each other, node $v$ becomes a child of $u$ when $v$ is isolated, or when both $u$ and $v$ are root nodes. Essentially, we allow different trees to get connected only when their roots interact (viewing isolated nodes as $1$-node trees).

\begin{table}[t]
\centering
\begin{tabular}{cc}
\noalign{\hrule height 1pt}\hline
rule & transition \\ 
\noalign{\hrule height 1pt}\hline
$\calS\calS$ & $(S,S) \rightarrow (R,L,\haschild)$ \\
$\calR\calS$ & $(R,S) \rightarrow (R,L,\haschild)$ \\
$\calI\calS$ & $(I,S) \rightarrow (I,L,\haschild)$ \\
$\calL\calS$ & $(L,S) \rightarrow (I,L,\haschild)$ \\
$\calR\calR$ & $(R,R) \rightarrow (R,I,\haschild)$ \\
\noalign{\hrule height 1pt}\hline
\end{tabular}
\caption{The interaction rules of {\bf TreeConstructor}. The rule $\calS\calS$ implies that when two isolated nodes (in state $S$) interact, they form a parent-child connection, one of them becomes the root (its state changes to $R$) of the newly formed $2$-node tree, and the other becomes the leaf (its state changes to $L$). Similarly, the rule $\calR\calS$ implies that when a root node (in state $R$) and an isolated node (in state $S$) interact, they form a parent-child connection, and the isolated node becomes a leaf (its state changes to $L$). The rest of the rules can be interpreted similarly.}
\label{TreeConstructor-rules}
\end{table}

\begin{figure}[ph!]
\centering

   \frame{\includegraphics[page=1,scale=1.1]{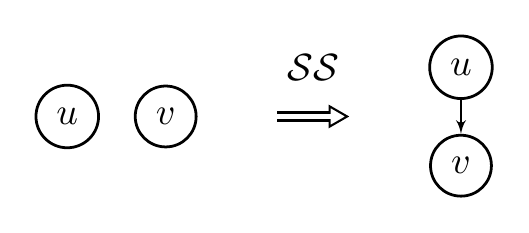}}
   \ \ \
   \frame{\includegraphics[page=2,scale=1.1]{tree_rules}}
   \vspace{8pt}
   
   \frame{\includegraphics[page=3,scale=1.1]{tree_rules}}
   \ \ \
   \frame{\includegraphics[page=4,scale=1.1]{tree_rules}}
   \vspace{8pt}
   
   \frame{\includegraphics[page=5,scale=1.1]{tree_rules}}
   \vspace{8pt}
   
   \frame{\includegraphics[page=6,scale=1.1]{tree_rules}}
   \ \ \
   \frame{\includegraphics[page=7,scale=1.1]{tree_rules}}

\caption{Graphical representation of the interaction rules used by the two protocols {\bf TreeConstructor} and {\bf BinaryTreeConstructor}.}
\label{fig:tree_rules}
\end{figure}

\begin{figure}[ph!]
\centering

   \frame{\includegraphics[page=1,scale=0.85]{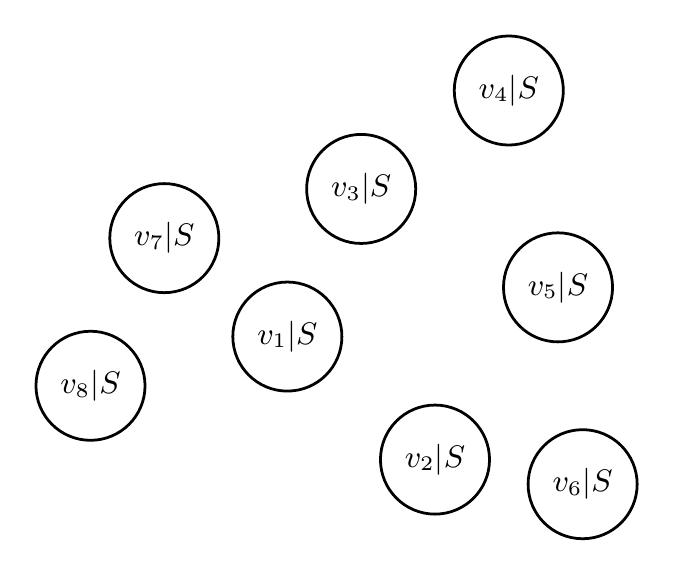}}
   \ \ \
   \frame{\includegraphics[page=2,scale=0.85]{example_arbitrary}}
   \vspace{8pt}
   
   \frame{\includegraphics[page=3,scale=0.85]{example_arbitrary}}
    \ \ \ 
   \frame{\includegraphics[page=4,scale=0.85]{example_arbitrary}}
   \vspace{8pt}
   
   \frame{\includegraphics[page=5,scale=0.85]{example_arbitrary}}
   \ \ \
   \frame{\includegraphics[page=6,scale=0.85]{example_arbitrary}}
   \vspace{8pt}
   
   \frame{\includegraphics[page=7,scale=0.85]{example_arbitrary}}
   \ \ \   
   \frame{\includegraphics[page=8,scale=0.85]{example_arbitrary}}

\caption{An example execution of the protocol {\bf TreeConstructor}. Each circle contains the name of the corresponding node and its current state. The ordering according to which the fair probabilistic scheduler selects the nodes that interact is as follows: $(v_1,v_2)$, $(v_3,v_4)$, $(v_1,v_4)$, $(v_3,v_5)$, $(v_1,v_2)$, $(v_1,v_3)$, $(v_2,v_6)$, $(v_7,v_8)$, $(v_1,v_7)$. Only the interactions that lead to new connections are depicted. See the relevant discussion in Example~\ref{example:abitrary}.}
\label{fig:example_arbitrary}
\end{figure}

Before we continue with the analysis of the protocol {\bf TreeConstructor}, we present an example of how it operates in order to construct an arbitrary tree network. 

\begin{example}\label{example:abitrary}
Consider an instance with $n=8$ nodes $\{v_1, ..., v_8\}$, and assume the following (incomplete) ordering over pairs of nodes, according to which the fair probabilistic scheduler selects the nodes that interact: $(v_1,v_2)$, $(v_3,v_4)$, $(v_1,v_4)$, $(v_3,v_5)$, $(v_1,v_2)$, $(v_1,v_3)$, $(v_2,v_6)$, $(v_7,v_8)$, $(v_1,v_7)$. Figure~\ref{fig:example_arbitrary} depicts how the arbitrary tree gets formed following the rules of {\bf TreeConstructor}. 
All nodes are initially isolated. We now describe the behavior of the protocol for each interaction:
\begin{itemize}
\item
When the nodes $(v_1,v_2)$ interact, they get connected according to rule $\calS\calS$: $v_2$ becomes a child of $v_1$, $v_1$ becomes a root node, and $v_2$ becomes a leaf node; since both nodes were isolated, the choice as to which of them becomes the parent is arbitrary in this case. The same happens at the second step for the pair $(v_3,v_4)$: $v_4$ becomes a child of $v_3$, $v_3$ becomes a root node, and $v_4$ becomes a leaf node. 

\item 
For the next pair $(v_3,v_5)$, since $v_3$ is a root node and $v_5$ is isolated, rule $\calR\calS$ gets invoked and the nodes get connected: $v_5$ becomes a child of $v_3$ and a leaf node.

\item  
When nodes $v_1$ and $v_4$ interact, nothing happens since $v_1$ is a root node and $v_4$ is a leaf node. {\bf TreeConstructor} has no rule for this case since such nodes should not get connected: $v_4$ already has a parent, and leaf nodes are not allowed to become parents of root nodes since such connections could create cycles when the interacting leaf and root nodes are already part of the same tree subnetwork. 

\item 
For the pair $(v_3,v_5)$ the rule $\calR\calS$ get invoked again and leads to $v_5$ becoming a child of $v_3$ and a leaf node, while for the pair $(v_1,v_2)$ nothing happens since they are already connected to each other. 

\item 
For the pair $(v_1,v_3)$, since both $v_1$ and $v_3$ are root nodes, the rule $\calR\calR$ get invoked and establishes a connection between them: $v_3$ becomes a child of $v_1$, $v_1$ remains a root node, and $v_3$ now becomes an internal node; again the choice as to which of them becomes the parent is arbitrary. 

\item
The remaining interactions and connections are similar: $(v_2,v_6)$ leads to $v_6$ becoming a child of $v_2$ and a leaf node; the interaction $(v_7,v_8)$ leads to $v_8$ becoming a child of $v_7$, $v_7$ becomes a root node, and $v_8$ becomes a leaf node; finally, the interaction $(v_1,v_7)$ leads to $v_7$ becoming a child of $v_1$ and an internal node.
\end{itemize}
Of course, with different choices for which node becomes the parent when possible, the tree network might end up having a different structure. Any other future interactions chosen by the scheduler cannot affect the structure since the formation process is already complete and there are no rules for the deletion of already established connections. 
\hfill \qed
\end{example}

Next, we show the correctness of the protocol and bound its time complexity. 

\begin{thm}
The protocol {\bf TreeConstructor} correctly constructs an arbitrary tree network, and its expected running time is $\Theta(n^2)$.
\end{thm}

\begin{proof}
To prove the correctness of {\bf TreeConstructor} we need to argue that 
\begin{itemize}
\item[(a)] no cycles will ever appear, and
\item[(b)] the network will eventually consist of a single connected component.
\end{itemize}
For (a), simply observe that the interaction rules of {\bf TreeConstructor} are such that {\em only different} components are allowed to get connected. Therefore, since initially all nodes are isolated (we can think of them as $1$-node trees), the only possible components that may form are trees. 
For (b), observe that at any step of time the network may consist of multiple components that are either isolated nodes or trees. Any interaction between an isolated node and any other node yields at least one less isolated node and at most one more root node, while any interaction between two root nodes yields one less root. We can think of the number of isolated nodes plus the number of roots as a potential which is initially equal to $n$ and decreases by exactly $1$ for every interaction involving isolated or root nodes; for any other type of interaction, its value remains unchanged. Hence, it eventually becomes equal to $1$, when there is only a single spanning tree left. 

For the running time, since we assume the existence of a fair probabilistic scheduler, every pair of nodes needs on average $\Theta(n^2)$ steps to interact and, therefore, during such a span of $\Theta(n^2)$ steps, each node will definitely interact with the node that will end up becoming its parent in the final tree network.
\end{proof}

\subsection{Binary trees}\label{sec:binary-tree}
For the formation of binary tree networks, we assume that each node $v$ is equipped with a register $w_v$ which initially stores a unique random number.\footnote{Assuming unique random numbers is only a simplification for the exposition of our results. In a real implementation, when an isolated node $v$ interacts with and connects as a child of another node $u$, we can set $w_u$ and $w_v$ equal to the current timestamp (which is unique) if $u$ is isolated as well, or we can set $w_v$ equal to $w_u$ in any other case like rule $\calU\calW$ suggests. This way, there is no randomness involved and everything is well-defined.} The values stored in these registers are used to define a monotonic merging of different components so that no cycles appear and the final outcome is a tree. 

\begin{table}[t]
\centering
\begin{tabular}{cc}
\noalign{\hrule height 1pt}\hline
rule & transition \\ 
\noalign{\hrule height 1pt}\hline
$\calS\calS$ & $(S,S) \rightarrow (R_1,L,\haschild)$ \\
$\calR\calS$ & $(R_1,S) \rightarrow (R_2,L,\haschild)$ \\
$\calI\calS$ & $(I_1,S) \rightarrow (I_2,L,\haschild)$ \\
$\calL\calS$ & $(L,S) \rightarrow (I_1,L,\haschild)$ \\
$\calR\calR$ & $(R_1,R_i,<) \rightarrow (R_2,I_i,\haschild)$  \\
$\calI\calR$ & $(I_1,R_i,<) \rightarrow (I_2,I_i,\haschild)$  \\
$\calL\calR$ & $(L,R_i,<) \rightarrow (I_1,I_i,\haschild)$ \\
$\calU\calW$ & $(u,v,\haschild) \rightarrow w_v := w_u$ \\
\noalign{\hrule height 1pt}\hline
\end{tabular}
\caption{The interaction rules of {\bf BinaryTreeConstructor}. For instance, the rule $\calI\calR$ implies when an internal node $u$ with one child (in state $I_1$) interacts with a root node $v$ with $i \in \{1,2\}$ children (in state $R_i$) and $w_u < w_v$, the form a parent-child connection, $u$ remains an internal node but now has two children (its state becomes $I_2$), and $v$ becomes an internal node with $i$ children (its state becomes $I_i)$. The other rules can be interpreted similarly.}
\label{BinaryTreeConstructor-rules}
\end{table}

Now, we consider the protocol {\bf BinaryTreeConstructor} which uses six states and the interaction rules of Table~\ref{BinaryTreeConstructor-rules}; see again Figure~\ref{fig:tree_rules} for a graphical description of the rules. Specifically, we say that a node $v$ is in state 
\begin{itemize}
\item $S$ if $v$ is isolated,
\item $L$ if $v$ is a leaf,
\item $I_i$ if $v$ is internal with $i\in \{1,2\}$ children, and 
\item $R_i$ if $v$ is a root with $i\in \{1,2\}$ children.
\end{itemize}
According to the interaction rules, when two nodes $u$ and $v$ interact with each other, node $v$ becomes a child of $u$ when $v$ is isolated, or when $v$ is a root and it holds that $w_u < w_v$. Essentially, by always respecting this inequality condition, we guarantee that, at any time, every distinct tree contains nodes that store values strictly greater than the value stored in the root of the tree. In this way, cycles that could appear due to the interactions between leaf or internal nodes with root nodes that are part of the same tree, are correctly avoided. For example, consider the interaction of a leaf node $v$ with the root $r$ of the tree that it belongs to. Then, since $w_v > w_r$, the condition of rule $\calL\calR$ is not satisfied and $r$ will not become a child of $v$. 

The interaction rule $\calU\calW$ is necessary to guarantee that the value stored in the root of each tree is spread to all other nodes of this tree so that a single spanning tree network can be formed. To see the necessity of this, consider two trees $T_1$ and $T_2$ with corresponding root nodes $r_1$ and $r_2$ (each with two children) such that $w_{r_1} < w_v$ for every $v \in T_2\setminus\{r_2\}$ and $w_{r_2} < w_v$ for every $v \in T_1\setminus\{r_1\}$. Then, without the interaction rule $\calU\calW$, these two trees will never get connected. However, with $\calU\calW$, it is now possible to spread the values $w_{r_1}$ and $w_{r_2}$ to all nodes in the two trees and, since one of them is strictly greater than the other, the connection of the trees will eventually become possible.

\begin{figure}[ph!]
\centering

   \frame{\includegraphics[page=1,scale=0.85]{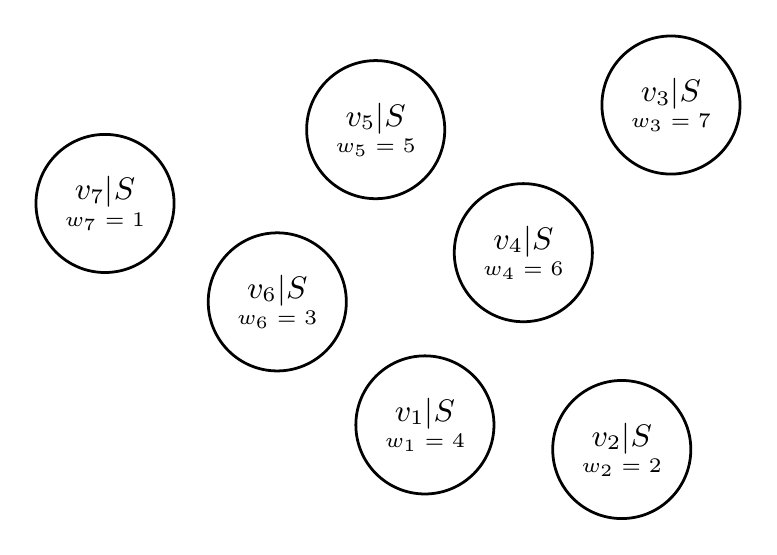}}
   \ \ \
   \frame{\includegraphics[page=2,scale=0.85]{example_binary}}
   \vspace{8pt}
   
   \frame{\includegraphics[page=3,scale=0.85]{example_binary}}
    \ \ \ 
   \frame{\includegraphics[page=4,scale=0.85]{example_binary}}
   \vspace{8pt}
   
   \frame{\includegraphics[page=5,scale=0.85]{example_binary}}
   \ \ \
   \frame{\includegraphics[page=6,scale=0.85]{example_binary}}
   \vspace{8pt}
   
   \frame{\includegraphics[page=7,scale=0.85]{example_binary}}
   \ \ \   
   \frame{\includegraphics[page=8,scale=0.85]{example_binary}}

\caption{An example execution of the protocol {\bf BinaryTreeConstructor}. Each circle contains the name of the corresponding node, its current state, and the value of the $w$-register. The ordering according to which the fair probabilistic scheduler selects the nodes that interact is as follows: $(v_1,v_2)$, $(v_3,v_4)$, $(v_7,v_6)$, $(v_1,v_4)$, $(v_5,v_6)$, $(v_2,v_7)$, $(v_3,v_4)$, $(v_1,v_6)$. Only the interactions that lead to new connections or change the value of the $w$-register are depicted. Notice that the interaction $(v_2,v_7)$ does not lead to $v_7$ becoming a child of $v_2$ since $w_2 > w_7$, while the second to last interaction $(v_3,v_4)$ only updates the value of $w_3$. See the relevant discussion in Example~\ref{example:binary}.}
\label{fig:example_binary}
\end{figure}

Let us present an example of how {\bf BinaryTreeConstructor} would operate in order to form a binary tree network.

\begin{example}\label{example:binary}
Consider an instance with $n=7$ nodes $\{v_1, ..., v_7\}$, and the following values for their corresponding $w$-registers: $w_1 = 4$, $w_2 = 2$, $w_3 = 7$, $w_4 = 6$, $w_5 = 5$, $w_6 = 3$, $w_7 = 1$.
Further, assume the following (incomplete) ordering over pairs of nodes, according to which the fair probabilistic scheduler selects the nodes that interact: $(v_1,v_2)$, $(v_3,v_4)$, $(v_7,v_6)$, $(v_1,v_4)$, $(v_5,v_6)$, $(v_2,v_7)$, $(v_3,v_4)$, $(v_1,v_6)$. Figure~\ref{fig:example_binary} depicts how the binary tree gets formed following the rules of {\bf BinaryTreeConstructor}. Initially all are isolated. Let us describe what happens in each of these interactions:
\begin{itemize}
\item
When the nodes $(v_1,v_2)$ interact, they get connected according to rule $\calS\calS$: $v_2$ becomes a child of $v_1$, $v_1$ becomes a root node with one child, $v_2$ becomes a leaf node, and $w_2$ is updated to be equal to $w_1$, that is, $w_2$ becomes equal to $1$; this is the result of rule $\calU\calW$, which we assume that is invoked immediately after a new connection has been established for the purposes of this example (in general, $\calU\calW$ can be invoked the next time the scheduler selects the same pair of nodes to interact).
Similar connections are established at the next two step for the pairs $(v_3,v_4)$ and $(v_7,v_6)$, and the $w$-registers are updated accordingly (see Figure~\ref{fig:example_binary}); recall that the choice as to which node becomes the parent when both nodes are isolated is arbitrary. 

\item 
Since both $v_1$ and $v_4$ are root nodes with one child at this point, the next interaction between them invokes rule $\calR\calR$, and the nodes get connected: $v_4$ becomes a child of $v_1$, $v_4$ becomes an internal node with one child, $v_1$ remains a root node with two children, and $w_4$ becomes equal to $w_1=4$; note that since both of them were root nodes with one child before their interaction and $w_1 = 4 < w_4=6$, $v_4$ can become a child of $v_1$, but not the other way around, even though they were both in the same state. 

\item 
For the next pair $(v_5,v_6)$, rule $\calL\calS$ get invoked and the nodes get connected: $v_5$ becomes a child of $v_6$, $v_5$ becomes a leaf nodes, $v_6$ becomes an internal node with one child (since it was already a leaf node), and $w_5$ becomes equal to $w_6=1$. 

\item
The next interaction between nodes $v_2$ and $v_7$ has no impact in the network structure: since $v_2$ is a leaf node at this point, it would only be possible for $v_7$ to become a child of $v_2$, but this does not happen since the $w$-registers do not respect the desired inequality (it should be $w_2 < w_7$, but $w_2 = 7$ and $w_7=1$).

\item
The interaction $(v_3,v_4)$ between two nodes that are already connected to each can only invoke rule $\calU\calW$ to update $w_3$ and set it equal to $w_4$; notice that $w_4$ was changed since $v_4$ became a child of $v_1$, but $w_3$ remained equal to the previous value it inherited from $w_4$ from the first interaction that led to their connection.

\item
Finally, the interaction $(v_1,v_6)$ invokes rule $\calI\calR$: $v_1$ becomes a child of $v_6$, both $v_1$ and $v_6$ now becomes internal nodes with two children, and $w_1$ becomes equal to $w_6$; observe here that $v_6$ could not become a child of $v_1$ since $v_1$ already has two children. 
\end{itemize}
Since the the binary tree network is now formed, any future interactions can only update the $w$-registers and at the end all of them will hold the same value, but this has not impact as the formation process has already been completed.
\hfill\qed
\end{example}

\begin{thm}\label{thm:binary}
The protocol {\bf BinaryTreeConstructor} correctly constructs a binary tree network, and its running time is $O(n^4)$.
\end{thm}

\begin{proof}
The correctness of the protocol follows by the above discussion: no cycles will ever appear, and all nodes will end up get connected in a single spanning binary tree.

For the running time, observe that for every possible merging of two different trees, it might be the case that we first need to update the values stored in all nodes of the two trees, and then wait for the right pair of nodes to interact in order for the merging to take place. This last step needs expected time $\Theta(n^2)$. For the update step, in the worst case, we must wait for all sequential parent-child interactions to happen, starting from the root and its children. Since there are $O(n)$ parent-child pairs in the tree and each such interaction takes average time $\Theta(n^2)$ to happen, the update of all values in a tree needs time $O(n^3)$. Therefore, every merging needs time $O(n^3) + O(n^2) = O(n^3)$, and since $n$ mergings must take place, the total time is $O(n^4)$.    
\end{proof}

\subsection{$k$-ary trees}\label{sec:k-ary-tree}
We now extend the ideas developed for binary tree networks in order to form $k$-ary tree networks, for any integer $k \geq 2$. 
Again we assume that each node $v$ is equipped with a register $w_v$ which initially stores a unique random number. 

Consider the protocol {\bf $k$-TreeConstructor} which uses $2k+2$ states and the $4k$ interaction rules of Table~\ref{kTreeConstructor-rules}. We say that a node $v$ is in state 
\begin{itemize}
\item $S$ if $v$ is isolated,
\item $L$ if $v$ is a leaf,
\item $I_i$ if $v$ is internal with $i\in \{1, ..., k\}$ children, and 
\item $R_i$ if $v$ is a root with $i\in \{1, ..., k\}$ children.
\end{itemize}
According to the interaction rules, when two nodes $u$ and $v$ interact with each other, node $v$ becomes a child of $u$ when $v$ is isolated, or when $v$ is a root, $u$ has strictly less than $k$ children already and it holds that $w_u < w_v$.  

\begin{table}[t]
\centering
\begin{tabular}{cl}
\noalign{\hrule height 1pt}\hline
rule & transition \\ 
\noalign{\hrule height 1pt}\hline
$\calS\calS$ & $(S,S) \rightarrow (R_1,L,\haschild)$ \\
$\calR\calS$ & $(R_1,S) \rightarrow (R_2,L,\haschild)$ \\
			 & ... \\
             & $(R_{k-1},S) \rightarrow (R_k,L,\haschild)$ \\			
$\calI\calS$ & $(I_1,S) \rightarrow (I_2,L,\haschild)$ \\
			 & ... \\
             & $(I_{k-1},S) \rightarrow (I_k,L,\haschild)$ \\			
$\calL\calS$ & $(L,S) \rightarrow (I_1,L,\haschild)$ \\
$\calR\calR$ & $(R_1,R_i,<) \rightarrow (R_2,I_i,\haschild)$  \\
			 & ... \\
             & $(R_{k-1},R_i,<) \rightarrow (R_k,I_i,\haschild)$ \\			
$\calI\calR$ & $(I_1,R_i,<) \rightarrow (I_2,I_i,\haschild)$  \\
			 & ... \\
             & $(I_{k-1},R_i,<) \rightarrow (I_k,I_i,\haschild)$ \\			
$\calL\calR$ & $(L,R_i,<) \rightarrow (I_1,I_i,\haschild)$ \\
$\calU\calW$ & $(u,v,\haschild) \rightarrow w_v := w_u$ \\
\noalign{\hrule height 1pt}\hline
\end{tabular}
\caption{The interaction rules of {\bf $k$-TreeConstructor}. The rules are partitioned into types depending on their function. For example, rules of type $\calI\calS$ deal with interactions among two nodes such that one node is internal and the other node is isolated; the different versions of rule take into account the different number of children the internal node may have.}
\label{kTreeConstructor-rules}
\end{table}

The proof of the next statement follows by arguments similar to those used for the proof of Theorem~\ref{thm:binary}. The running time remains asymptotically the same since again it might be the case that the values stored by roots of different trees may need to be diffused to all nodes of these trees (in order for their merging to happen); this scenario is clearly not affected by the number of children a node may have.  

\begin{thm}\label{thm:k}
Given $k\geq 2$, the protocol {\bf $k$-TreeConstructor} correctly constructs a $k$-ary tree network, and its running time is $O(n^4)$.
\end{thm}

\subsection{Depth and tree height estimation}\label{sec:depth-height}
As we will see later, information about the depth of each node and the height (max depth) of the tree can be very useful, especially since this information can be computed locally by each node through its interactions with other nodes of the network. We assume that each node $v$ is equipped with two more registers $d_v$ and $h_v$, which store the node's estimation about its own depth and the height of the tree, respectively. Both of them are initially equal to zero, and are updated according to the two interaction rules of Table~\ref{depth-height}.

\begin{table}[t]
\centering
\begin{tabular}{cc}
\noalign{\hrule height 1pt}\hline
rule & transition \\ 
\noalign{\hrule height 1pt}\hline
$\calU\calD$ & $(u,v,\haschild) \rightarrow d_v := d_u + 1 $ \\
$\calU\calH$ & $(u,v) \rightarrow h_u :=  h_v := \max\{h_u,h_v,d_u,d_v \}$ \\
\noalign{\hrule height 1pt}\hline
\end{tabular}
\caption{Depth and tree height estimation.}
\label{depth-height}
\end{table}

Of course, while the network is under construction, every node $v$ has a wrong estimation about these quantities, but eventually after the completion of the network structure, the values stored in the registers $d_v$ and $h_v$ will stabilize to the correct ones such that $d_r=0$ only for the root node $r$ of the tree, and $h_v = \max_u d_u = h$ for every node $v$. Next, we formally argue about the correctness of rules $\calU\calD$ and $\calU\calH$.

\begin{prop}
Rules $\calU\calD$ and $\calU\calH$ correctly compute the depth of each node and the height of the tree. 
\end{prop}

\begin{proof}
Since the register $d_v$ is initially equal to zero for every node $v$, once the network is constructed, the root node $r$ of the tree will definitely have the correct estimation about its own depth, that is, $d_r=0$; notice that $r$ changes its state only once (from isolated to root) and never updates the register value of $d_r$. Therefore, rule $\calU\calD$ will correctly compute the depth of the root's children when these nodes interact with $r$, and the information will sequentially be spread in an up-bottom manner to all nodes of the network. Then, once every node correctly knows its own depth, rule $\calU\calH$ will also spread the information about the maximum depth in a bottom-up manner to all nodes. 
\end{proof}

\section{Energy redistribution protocols}\label{sec:energy-protocols}
In this section, we present several energy redistribution protocols and discuss some of their theoretical properties, mainly focusing on the lossless case. In our simulations in the next section, we will also consider the more general lossy case, and evaluate the redistribution protocols presented in the following. 

To study whether it is possible to achieve good energy distributions, we present protocols which make a plethora of different assumptions about the information of the network that the nodes can exploit. We start by assuming that the nodes have complete knowledge about the structure, the total energy, and the number of nodes in the network. This allows us to design a very powerful protocol that can achieve an exact distribution in the lossless case. Then, we consider the other extreme case in which the nodes have no global information, and show that such protocols do not  convergence to good distributions in general, but do so in many fundamental special cases, such as when the network happens to be a line. Finally, we also consider an intermediate scenario according to which the nodes have partial global information about the total initial energy in the network, but know nothing about its structure. 

\subsection{Exploiting global network information}\label{sec:ideal-target}
We start with the presentation of a protocol that exploits the whole network structure; we refer to it as {\bf ideal-target}. In our simulations in the next section, the performance of {\bf ideal-target} will serve as an upper bound on the performance of the rest of our protocols, which we will define next. 

Given a tree structure $T$ and the total energy of the network $E_\total$, it is easy to compute the {\em ideal} energy $\gamma_v(T,E_\total)$ that each node $v$ must have so that the energy distribution of the network is exact. Let $h$ denote the height (maximum depth) of the tree, $d_v \in \{0, ..., h\}$ denotes the depth of node $v$, and $n_d$ is the number of nodes at depth $d \in \{0, ..., h\}$. By the definition of the exact energy distribution, all nodes at depth $d < h$  must have equal ideal energy that is twice the energy of every node at depth $d+1$. Let $x$ denote the ideal energy of any node at depth $h$. Then, the ideal energy of node $v$ is $\gamma_v(T,E_\total) = 2^{h-d_v} \cdot x$, where 
$$x = \frac{E_\total}{\sum_{d=0}^h n_d \cdot 2^{h-d}}$$
so that $\sum_v \gamma_v(T,E_\total) = E_\total$.\footnote{Observe that if the height of the tree is large, then nodes of higher depth will have an ideal energy that is close to zero. The worst case is when the tree ends up being a line.}

Now, the protocol {\bf ideal-target} is defined as follows: given the tree network structure $T$ (that is, the height of the tree $h$, the depth $d_v$ of every node $v \in V$, and the number of nodes $n_d$ at every depth $d \in \{0, ..., h\}$) and the initial total energy $E_\total$ of the network, simply compute the ideal energy $\gamma_v(T,E_\total)$ for every node $v$, and then use this as the {\em target} energy that $v$ must end up storing. So, node $v$ asks for energy if its current energy is below the target, and gives away energy when its current energy is exceeding the target; see Protocol~\ref{protocol:ideal-target}.

\begin{algorithm}[t]
\SetAlgoLined
\KwIn{interacting pair of nodes $(u,v)$ at time step $t$, energy loss $\beta$}
\KwOut{updated energies $E_u(t+1)$ and $E_v(t+1)$}

    \uIf {$E_u(t) > T_u(t)$ and $E_v(t) < T_v(t)$}{
  		$x := \min\{ E_u(t)-T_u(t), T_v(t) - E_v(t) \}$ \\
		$E_u(t+1) := E_u(t) - x$ \\
		$E_v(t+1) := E_v(t) + (1-\beta) x$ 
   	}	
   	\ElseIf {$E_u(t) < T_u(t)$ and $E_v(t) > T_v(t)$}{
   	    $x := \min\{ T_u(t) - E_u(t), E_v(t) - T_v(t) \}$ \\
		$E_u(t+1) := E_u(t) + (1-\beta) x$ \\
		$E_v(t+1) := E_v(t) - x$ 
    }
   	
\caption{{\bf ideal-target}}
\label{protocol:ideal-target}
\end{algorithm}

\begin{figure}[t!]
\centering
   \frame{\includegraphics[scale=1.1]{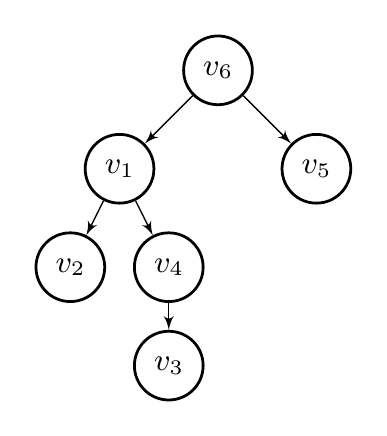}}
\caption{The tree network used in Examples~\ref{example:ideal}--\ref{example:depth-target}, which showcase the operation of the energy redistribution protocols.}
\label{fig:example-distribution}
\end{figure}

\begin{example}\label{example:ideal}
Let $\beta=0$ and consider an instance with $n=6$ nodes, which are connected according to the binary tree network $T$ that is depicted in Figure~\ref{fig:example-distribution}. The total energy in the network is $E_\total=2100$, and the amount of energy that each node currently holds is as follows: $E_1 = 500$, $E_2 = 100$, $E_3 = 150$, $E_4 = 400$, $E_5 = 350$, $E_6 = 600$. The tree has height $h=3$, and the number of nodes per depth level is as follows: $n_0=1$, $n_1=2$, $n_2=2$, $n_3=1$. Therefore, we can compute the ideal energy of every node. Since $\sum_{d=0}^h n_d \cdot 2^{h-d} = 1 \cdot 2^{3-0} + 2 \cdot 2^{3-1} + 2 \cdot 2^{3-2} + 1 \cdot 2^{3-3} = 21$, we have that $x = 100$. Consequently, the ideal energies of the nodes are as follows: $\gamma_{v_6}(T,E_\total) = 2^{3-0} \cdot 100 = 800$, $\gamma_{v_1}(T,E_\total) = \gamma_{v_5}(T,E_\total) = 2^{3-1} \cdot 100 = 400$, $\gamma_{v_2}(T,E_\total) = \gamma_{v_4}(T,E_\total) = 2^{3-2} \cdot 100 = 200$, and $\gamma_{v_3}(T,E_\total) = 2^{3-3} \cdot 100 =100$.

Now assume that the next interaction is between $v_1$ and $v_2$, which form a parent-child pair. Since $v_1$ has $500$ units of energy and its ideal energy is $400$, it would like to give away the $100$ extra units of energy that it has. On the other hand, since $v_2$ currently has $150$ units of energy and its ideal is $200$, it would like to get $50$ more units of energy to reach its goal. Hence, it is possible for $v_1$ to give these $50$ units to $v_2$. As a result, $v_2$ has now reached its ideal energy and $v_1$ is by $50$ units closer to its own ideal. The protocol continues similarly with other parent-child interactions until every node has reached its ideal energy. 
\hfill\qed
\end{example}

\begin{thm}
The protocol {\bf ideal-target} converges to an exact energy distribution for $\beta=0$.
\end{thm}

\begin{proof}
Observe that for $\beta=0$, a given tree structure $T$ and total network energy $E_\total$, the exact energy distribution is {\em unique}. Hence, by the definition of {\bf ideal-target}, where each node has a target energy that is equal to the energy it has in the unique exact energy distribution, the network will converge to exactly this energy distribution.
\end{proof}

Unfortunately, for $\beta \neq 0$, the total energy of the network is continuously decreasing from interaction to interaction. Therefore, we cannot guarantee that the network will converge to an exact energy distribution. In fact, it might be the case that the outcome energy distribution is not even relaxed; for instance, the root may end up with less than twice the energy of its children. However, as we will see later, {\bf ideal-target} converges to a stable energy distribution fast enough so that the total energy loss is small and the final energy distribution is close to the exact energy distribution (as computed using the initial total energy of the network).

\subsection{Oblivious protocols}
Even though {\bf ideal-target} looks like the perfect protocol, it assumes that the nodes can infer global information about the network structure. Since the nodes are computationally weak devices with limited memory, they cannot ``remember'' how many nodes there are in every different depth level nor have a correct estimation of the total energy in the network (let alone the initial total energy). So, we now turn our attention in defining two protocols that are oblivious to the structure of the network and need much less information to operate. These protocols redistribute energy only when parent-child pairs of nodes interact.

For any real $\lambda \geq 2$, the protocol {\bf $\lambda$-exchange} requires that when a parent-child pair of nodes $(p,c)$ interacts, their energy is redistributed so that $p$ has exactly $\lambda$ times the energy of $c$ {\em only if} $p$ has initially strictly less than $\lambda$ times the energy of $c$; see Protocol~\ref{protocol:lambda-exchange}. For any real $\kappa \in (0,1)$, we also consider the protocol {\bf $\kappa$-transfer}, which requires that when a parent-child pair $(p,c)$ interacts and $p$ has less than twice the energy of $c$, then $c$ transfers a $\kappa$-fraction of its energy to $p$; see Protocol~\ref{protocol:lambda-transfer}. 

\begin{example}\label{example:oblivious}
Consider again the instance with $n=6$ nodes that are connected according to the network that is depicted in Figure~\ref{fig:example-distribution}, and $\beta=0$. Recall that the amounts of energy that the nodes have are $E_1 = 500$, $E_2 = 100$, $E_3 = 150$, $E_4 = 400$, $E_5 = 350$, $E_6 = 600$. Let us briefly examine how {\bf $2$-exchange} and {\bf $0.5$-transfer} operate for two parent-child interactions $(v_1,v_2)$ and $(v_1,v_4)$. 

Since $E_1 > 2 E_2$, both protocols will not change anything when $v_1$ interacts with $v_2$. However, since $E_1 < 2 E_4$, when $v_1$ interacts with $v_4$, {\bf $2$-exchange} gets $x = \frac{1}{3}(2 \cdot 400 - 500) = 100$ units of energy from $v_4$ and gives them to $v_1$ so that $v_1$ now has $600$ units of energy, while $v_4$ has $300$; this way $v_1$ who is the parent node has exactly twice the energy of its child $v_4$. On the other hand, {\bf $0.5$-transfer} will transfer exactly have of $v_4$'s energy to $v_1$ so that $v_1$ has $700$ units of energy and $v_4$ has $200$. 
\hfill\qed
\end{example}

One can easily observe that these protocols are designed to achieve a relaxed energy distribution rather than an exact one. In order to have any chance to converge to an exact energy distribution, we need to change the condition under which energy exchanges take place from $p$ having energy strictly less than $\lambda$ times (for {\bf $\lambda$-exchange}) or twice (for {\bf $\kappa$-transfer}) the energy of $c$ to $p$ not having {\em exactly} twice the energy of $c$ (for both protocols); notice that {\bf $\lambda$-exchange} makes sense only for $\lambda=2$ in this case. We refer to the latter condition as the exact equilibrium condition. 

\begin{algorithm}[t]
\SetAlgoLined
\KwIn{interacting parent-child pair $(p,c)$ at time step $t$, energy loss $\beta$}
\KwOut{updated energies $E_p(t+1)$ and $E_c(t+1)$}
	\If {$E_p(t) < \lambda E_c(t)$} {	
	    $x := \frac{1}{\lambda+1}\bigg( \lambda E_c(t) - E_p(t) \bigg)$ \\
	    $E_p(t+1) := E_p(t) + (1-\beta) x$ \\
	    $E_c(t+1) := E_c(t) - x$	 
	}
\caption{{\bf $\lambda$-exchange}}
\label{protocol:lambda-exchange}
\end{algorithm}

\begin{algorithm}[t]
\SetAlgoLined
\KwIn{interacting parent-child pair $(p,c)$ at time step $t$, energy loss $\beta$}
\KwOut{updated energies $E_p(t+1)$ and $E_c(t+1)$}
	\If {$E_p(t) < 2E_c(t)$} {	
	    $E_p(t+1) := E_p(t) + (1-\beta) \kappa E_c(t)$ \\
	    $E_c(t+1) := (1-\kappa)E_c(t)$	 
	}
\caption{{\bf $\kappa$-transfer}}
\label{protocol:lambda-transfer}
\end{algorithm} 

\begin{thm}
The protocols {\bf $2$-exchange} and {\bf $\kappa$-transfer} may not converge to the exact energy distribution, even when they are defined using the exact equilibrium condition and $\beta=0$.
\end{thm}

\begin{proof}
To show the theorem it suffices to show explicit instances for which the protocols do not converge to the exact energy distribution. To this end, consider an instance with three nodes $\{a, b, c\}$ and assume that the probabilistic fair scheduler that processes all possible pairs of nodes in the order $\{(a,b)$, $(b,c)$, $(a,c)\}$ in a round-robin manner. Following this order of interactions, any of the tree formation protocols presented in the previous section may form the line-tree network $a \haschild b \haschild c$, where $a$ is the root, $b$ is an internal node, and $c$ is a leaf. Since interactions between nodes $a$ and $c$ do not affect the energy distribution ($a$ and $c$ do not have a parent-child connection, which is required by both {\bf $2$-exchange} and {\bf $\kappa$-transfer}), we can assume without loss of generality that during any step of time the only interactions that take place are between the pairs $(a,b)$ and $(b,c)$. 

For {\bf $2$-exchange}, consider the case where all nodes have the same initial energy $E_a(0) = E_b(0) = E_c(0)$. Let $t \geq 1$ be a step of time such that both parent-child pairs $(a,b)$ and $(b,c)$ are in exact energy equilibrium, and therefore the network has converged to the exact energy distribution. We will examine the case where the interaction at time $t$ is between the parent-child pair $(b,c)$; the other case of $(a,b)$ is obviously symmetric. During the previous time step $t-1$, due to the interaction between $a$ and $b$, we have that $(a,b)$ must be in exact energy equilibrium: $E_a(t-1) = 2 E_b(t-1)$. At time $t$, since the pair $(a,b)$ remains in equilibrium, we have that
$E_a(t) = 2E_b(t)$. Also, since node $a$ does not participate in the scheduled interaction at time $t$, we obtain that $E_a(t) = E_a(t-1)$. Combining these three relations, yields that $E_b(t) = E_b(t-1)$.
Since there is no energy loss ($\beta=0$), the condition $E_a(t) + E_b(t) + E_c(t) = E_a(t-1) + E_b(t-1) + E_c(t-1)$ yields that $E_c(t) =  E_c(t-1)$. Therefore, if there exists a time step $t$ during which the energy distribution of the network is exact, it must be the case that at time $t-1$ the energy distribution of the network is also exact. Recursively, this requires that the energy distribution of the network is exact from the beginning, which is not true. Hence, there cannot be any such $t$.

For {\bf $\kappa$-transfer}, consider the case where the nodes have initial energy such that $E_a(0) > 2[E_b(0) + E_c(0)]$. By the definition of the protocol, node $a$ will never transfer energy to the other nodes. Hence, due to the particular initial energy, there exists no time step $t$ such that $E_a(t)=2E_b(t)$. The theorem follows.
\end{proof}

Actually, we believe that no oblivious protocol can converge to an exact energy distribution, even for $\beta=0$. We were not able to prove such a statement, but our intuition is that fixing a parent-child pair of nodes $(p,c)$ in exact energy equilibrium will inevitably break other pairs of nodes that are already in equilibrium.

Next, we will prove that if there is no energy loss ($\beta=0$) and the tree network is a line (each node has at most one child), then {\bf $\lambda$-exchange} always converges to some relaxed distribution, for any $\lambda \geq 2$. Our simulations in the next section indicates that this is true for any kind of tree network, and that this is true for {\bf $\kappa$-exchange} as well. However, as we will see, there is a huge difference between these protocols in terms of the quality of the final energy distribution.  

\begin{thm}
When $\beta=0$ and the tree network is a line, the protocol {\bf $\lambda$-exchange} always converges to a relaxed energy distribution for any $\lambda \geq 2$.
\end{thm}

\begin{proof}
Let $v_1 \haschild v_2 \haschild ... \haschild v_n$ be the line-tree network: $v_1$ is the root, nodes $v_i$ for $i\in \{2, ..., n-1\}$ are internal, and $v_n$ is the leaf. To simplify our notation in what follows, we denote by $E_i(t)$ the energy that node $v_i$ has at time $t$. We will define a function $\Phi: \mathbb{N}_{\geq 0} \rightarrow \mathbb{R}_{\geq 0}$ with the following three properties:
\begin{itemize}
\item[(P1)] $\Phi(t)=0$ means that the network converges to a relaxed energy distribution at time $t$,
\item[(P2)] $\Phi(t)$ is non-increasing for every $t \geq 0$, and
\item[(P3)] there exists a number of time steps during which $\Phi$ strictly decreases until it reaches $0$.
\end{itemize}
Then, $\Phi$ is a potential function. Essentially, properties (P2) and (P3) guarantee that $\Phi$ will reach its global minimum value of zero, and property (P1) guarantees that when $\Phi$ has reached this global minimum value, the network has converged to a relaxed energy distribution.

To this end, for each $i\in \{1, ..., n-1\}$, let $x_i(t) = \one\{E_i(t) < \lambda E_{i+1}(t)\}$ be the binary variable indicating whether at time step $t$ the energy $E_i(t)$ of node $v_i$ is strictly less than $\lambda$ times the energy $E_{i+1}(t)$ of its child node $v_{i+1}$ or not, and consider the function 
\begin{align}\label{eq:potential}
\Phi(t) = \sum_{i=1}^{n-1} \bigg( \lambda E_{i+1}(t) - E_i(t) \bigg) \cdot x_i(t).
\end{align}
Observe that $\Phi(t) \geq 0$ for every $t$. Moreover, property (P1) is satisfied since if $\Phi(t)=0$ for some $t$, then $E_i(t) \geq \lambda E_{i+1}(t) \geq 2E_{i+1}(t)$ for every parent-child pair $(v_i,v_{i+1})$. So, it remains to prove that $\Phi$ satisfies properties (P2) and (P3) as well.

For property (P2), let $t \geq 1$ be a time step during which a parent-child pair $(v_{i^*},v_{i^*+1})$ with $x_{i^*}(t-1)=1$ interacts; otherwise, obviously there is no change during step $t$ and $\Phi(t)=\Phi(t-1)$. 
Then, the {\bf $\lambda$-exchange} protocol will set 
$$
E_{i^*}(t) = \frac{\lambda}{\lambda+1}\bigg(E_{i^*}(t-1) + E_{i^*+1}(t-1) \bigg)
$$
and 
$$
E_{i^*+1}(t) = \frac{1}{\lambda+1}\bigg(E_{i^*}(t-1) + E_{i^*+1}(t-1) \bigg)
$$
so that $E_{i^*}(t) = \lambda E_{i^*+1}(t)$, and therefore $x_{i^*}(t)=0$. 
To prove that $\Phi$ is non-increasing, we have to show that $\Phi(t) - \Phi(t-1) \leq 0$.
We distinguish between cases, depending on the states of $v_{i^*}$ and $v_{i^*+1}$. 

First, we consider the case $i^*=1$: $v_1$ is the root of the tree and $v_2$ is internal; the case $i^*=n-1$ is symmetric since we can think of node $v_n$ as the root node instead (it is only a matter of perspective). 
Since $x_i(t)=x_i(t-1)$ for every $i \not\in \{1,2\}$, by \eqref{eq:potential} (definition of the function $\Phi$) we have that
\begin{align}\nonumber
\Phi(t) - \Phi(t-1) &= \sum_{i=1}^{n-1} \bigg( \lambda E_{i+1}(t) - E_i(t) \bigg) \cdot x_i(t)
- \sum_{i=1}^{n-1} \bigg( \lambda E_{i+1}(t-1) - E_i(t-1) \bigg) \cdot x_i(t-1)  \\ \nonumber
&= \bigg( \lambda E_2(t) - E_1(t) \bigg) x_1(t) - \bigg( \lambda E_2(t-1) - E_1(t-1) \bigg) x_1(t-1) \\ \label{eq:potential-1}
&\quad + \bigg( \lambda E_3(t) - E_2(t) \bigg) x_2(t) - \bigg( \lambda E_3(t-1) - E_2(t-1) \bigg) x_2(t-1).
\end{align}
Since $x_1(t-1) = 1$ and $x_1(t) = 0$, observe that the worst case occurs when $x_2(t-1) = 0$ and $x_2(t) = 1$ so that the last expression in the above sequence of equalities is maximized; that is, the relaxed energy equilibrium of nodes $(v_1,v_2)$ breaks the equilibrium of nodes $(v_2,v_3)$. Therefore, since $v_3$ does not participate in the interaction of time $t$, we have that $E_2(t-1) \geq \lambda E_3(t-1) = \lambda E_3(t)$, and \eqref{eq:potential-1} becomes
\begin{align*}
\Phi(t) - \Phi(t-1) &= E_1(t-1) - \lambda E_2(t-1) + \lambda E_3(t) - E_2(t) \\
&\leq E_1(t-1) - (\lambda - 1) E_2(t-1) - \frac{1}{\lambda+1} \bigg(E_1(t-1) + E_2(t-1) \bigg) \\
&= \frac{\lambda}{\lambda+1} \bigg(E_1(t-1) - \lambda E_2(t-1) \bigg) \\
&< 0,
\end{align*}
as desired. The last inequality follows since $x_1(t-1)=1$. 

Next, we consider the case $1 < i^* < n-1$.
Similarly to before, $x_i(t)=x_i(t-1)$ for each $i \not\in \{i^*-1,i^*,i^*+1\}$. 
Hence, by \eqref{eq:potential}, we have that
\begin{align} \nonumber
\Phi(t) - \Phi(t-1) 
&= \bigg( \lambda E_{i^*}(t) - E_{i^*-1}(t) \bigg) x_{i^*-1}(t) - \bigg( \lambda E_{i^*}(t-1) - E_{i^*-1}(t-1) \bigg) x_{i^*-1}(t-1) \\ \nonumber
& \quad + \bigg( \lambda E_{i^*+1}(t) - E_{i^*}(t) \bigg) x_{i^*}(t) - \bigg( \lambda E_{i^*+1}(t-1) - E_{i^*}(t-1) \bigg) x_{i^*}(t-1) \\ \label{eq:potential-2}
& \quad + \bigg( \lambda E_{i^*+2}(t) - E_{i^*+1}(t) \bigg) x_{i^*+1}(t) - \bigg( \lambda E_{i^*+2}(t-1) - E_{i^*+1}(t-1) \bigg) x_{i^*+1}(t-1). 
\end{align}
Given that $x_{i^*}(t-1) = 1$ and $x_{i^*}(t) = 0$, the worst case occurs when $x_{i^*-1}(t-1) = 0$, $x_{i^*-1}(t) = 1$, $x_{i^*+1}(t-1) = 0$ and $x_{i^*+1}(t) = 1$ so that the last expression is maximized; essentially, in the worst case the relaxed energy equilibrium of nodes $(v_{i^*},v_{i^*+1})$  leads to breaking the equilibrium of both pairs of nodes $(v_{i^*-1},v_{i^*})$ and $(v_{i^*+1},v_{i^*+2})$. Therefore, since nodes $v_{i^*-1}$ and $v_{i^*+2}$ do not participate in the interaction of time $t$, we have that $E_{i^*-1}(t) = E_{i^*-1}(t-1) \geq \lambda E_{i^*}(t-1)$ and $E_{i^*+1}(t-1) \geq \lambda E_{i^*+2}(t-1) = \lambda E_{i^*+2}(t)$. Also, since $\lambda E_{i^*}(t) - E_{i^*+1}(t) = (\lambda-1) \bigg( E_{i^*}(t-1) + E_{i^*+1}(t-1) \bigg)$, \eqref{eq:potential-2} becomes 
\begin{align*}
\Phi(t) - \Phi(t-1) = \lambda E_{i^*}(t) - E_{i^*-1}(t) -  \lambda E_{i^*+1}(t-1) + E_{i^*}(t-1) + \lambda E_{i^*+2}(t) - E_{i^*+1}(t) \leq 0.
\end{align*}

For property (P3), it suffices to show that if $\Phi(t)\neq 0$ for some time step $t$ and $x_1(t)=x_{n-1}(t)=0$, then there exists a sequence of interactions between internal nodes that leads to $x_1(\tau)=1$ or $x_{n-1}(\tau)=1$ for some time step $\tau > t$. If this is true, then the function $\Phi$ will eventually (after multiple interactions at the two endpoints of the line) reach its minimum value of $0$, as it strictly decreases when interactions between nodes $(v_1,v_2)$ and $(v_{n-1},v_n)$ take place and change the energy of the participating nodes; any other parent-child pair interaction does not increase the value of $\Phi$. 

Now, assume a time step with the above specifications. Since $\Phi(t)\neq 0$, there exists an $i \in \{1, ..., n-1\}$  such that $x_i(t)=1$ and $x_\ell(t)=0$ for $\ell < i$, as well as a $j \in \{1, ..., n-1\}$ such that $x_j(t)=1$ and $x_\ell(t)=0$ for $\ell > j$. In the worst case, after the interaction of nodes $(v_i,v_{i+1})$ or $(v_j,v_{j+1})$ at some time step $t_i > t$ or $t_j > t$, it will be $x_{i-1}(t_i)=1$ or $x_{j+1}(t_j)=1$. Inductively, since the energy is in general flowing upwards in the network, such interactions will define chains of future interactions that can lead to at least one of the pairs $(v_1,v_2)$ or $(v_{n-1},v_n)$ to no longer satisfy the condition of {\bf $\lambda$-exchange}, i.e., there exists a time step $\tau$ such that $x_1(\tau)=1$ or $x_{n-1}(\tau)=1$. The proof is now complete.
\end{proof}

We remark that the above theorem guarantees {\em only} convergence to a relaxed energy distribution. However, it is not obvious that this convergence will happen quickly, or that the final relaxed distribution will be close to an exact one. We expect that the convergence time decreases as $\lambda$ increases since the latter increases the energy that flows upwards. On the other hand, the distance of the outcome energy distribution from the exact energy distribution (which can be thought of as the ``optimal'' distribution) increases as $\lambda$ increases since nodes at lower depth will end up with substantially more energy than those in higher depth; in fact, for high enough $\lambda$, the root might concentrate almost all of the network energy. Hence, it is important to fine-tune the parameter $\lambda$ and balance these trade-offs; we do this in the next section where we present our simulations.

\subsection{Mixing locally inferred and global information}\label{sec:apx-distribution}
So far, we have made two extreme informational assumptions: the nodes either have complete global information, or no information at all. We complement these two by further considering an intermediate setting according to which the nodes have access to only to {\em some} information about the network (such as the total initial energy, but not the network structure). For $k$-ary tree networks (with integer $k \geq 2$), we define the protocol {\bf $k$-depth-target} which sets a target energy for each node by exploiting some network information that can be locally estimated as well as some global network information that is provided as input to the nodes. Specifically, during each time step $t$, every non-root node $v$ sets a target energy 
\begin{align}\label{eq:zeta}
\zeta_v(t) = \frac{E_\total}{k^{d_v(t)}(h_v(t) + 1)},
\end{align}
where $d_v(t)$ is the estimation of node $v$ for its own depth at time $t$, $h_v(t)$ is the estimation of $v$ for the height of the tree at time $t$, and $E_\total$ is the initial total network energy. As already discussed in Section~\ref{sec:depth-height}, once the network is formed, the estimations $d_v(t)$ and $h_v(t)$ stabilize to the correct values. Then, the target is set to the correct value as well, and eventually all non-root nodes will end up storing exactly this much energy, while the root $r$ will collect the remaining $E_\total - \sum_{v \neq r}\zeta_v(t)$ units of energy. A description of {\bf $k$-depth-target} is given as Protocol~\ref{protocol:depth-target}.

Before we continue, we show that the target energy of each node is feasible when $\beta=0$ in the sense that the total target energy of non-root nodes is at most the total energy of the network. To see this, observe that when the depth and height estimations of the nodes stabilize to the correct ones, then all nodes of the same depth $d$ will have the same target energy $\frac{E_\total}{k^{d}(h + 1)}$, where $h$ is the true height of the $k$-ary tree. Hence, if $n_d$ denotes the number of nodes at depth $d$, then $n_d \leq k^d$ (since each node has at most $k$ children) and the total target energy of non-root nodes is 
\begin{align*}
\sum_{d=1}^h n_d \frac{E_\total}{k^{d}(h + 1)} \leq \sum_{d=1}^h \frac{E_\total}{h + 1} = \frac{h \cdot E_\total}{h+1} < E_\total. 
\end{align*}

\begin{algorithm}[t]
\SetAlgoLined
\KwIn{interacting pair of nodes $(u,v)$ at time step $t$, energy loss $\beta$}
\KwOut{updated energies $E_u(t+1)$ and $E_v(t+1)$}

    \uIf {$u,v \not\in R$ and $E_u(t) > \zeta_u(t)$ and $E_v(t) < \zeta_v(t)$}{
  		$x := \min\{ E_u(t)-\zeta_u(t), \zeta_v(t)-E_v(t) \}$ \\
		$E_u(t+1) := E_u(t) - x$ \\
		$E_v(t+1) := E_v(t) + (1-\beta)x$ \\
   	}
   	\uElseIf {$u \in R$}{
   	    $y:= \zeta_v(t) - E_v(t)$ \\
        $x := \min\{ |y|, E_u(t) \}$ \\
        $E_u(t+1) := E_u(t) + \sgn\{y\} x - \beta x \cdot \one\{y>0\}$ \\
        $E_v(t+1) := E_v(t) - \sgn\{y\} x - \beta x \cdot \one\{y<0\}$ \\
   	}	
   	\ElseIf {$v \in R$}{
   	    $y:= \zeta_u(t) - E_u(t)$ \\
        $x := \min\{ |y|, E_v(t) \}$ \\
        $E_u(t+1) := E_u(t) - \sgn\{y\} x - \beta x \cdot \one\{y<0\}$ \\
        $E_v(t+1) := E_v(t) + \sgn\{y\} x - \beta x \cdot \one\{y>0\}$ \\
   	}
  
\caption{{\bf $k$-depth-target}}
\label{protocol:depth-target}
\end{algorithm} 

\begin{example}\label{example:depth-target}
Consider yet again the instance with $n=6$ nodes that are connected according to the network that is depicted in Figure~\ref{fig:example-distribution}, and $\beta=0$. Recall that the total energy in the network is $E_\total$ and is distributed to the nodes as follows: $E_1 = 500$, $E_2 = 100$, $E_3 = 150$, $E_4 = 400$, $E_5 = 350$, $E_6 = 600$. We will focus on the protocol {\bf $2$-depth-target}. Since the height of the tree is $h=3$, we can compute the target energy of each non-root node using \eqref{eq:zeta} as follows: 
$\zeta_{v_1} = \zeta_{v_5} = \frac{2100}{2^1(3+1)} = 262.5$, 
$\zeta_{v_2} = \zeta_{v_4} = \frac{2100}{2^2(3+1)} = 131.25$, 
and $\zeta_{v_3} = \frac{2100}{2^3(3+1)} = 65.625$.
Observe that since the tree network has been formed, these target energies are not just estimations, they are the final targets that the nodes have, which is why they do not depend on the current time step.  
Let us see how the protocol operates for two parent-child interactions $(v_1,v_2)$ and $(v_1,v_6)$.

Since $v_1$ has $500$ units of energy and its target is $265.5$, it can give the $31.25$ units of energy that $v_2$ would like to get in order to reach its own target of $131.25$ from the $100$ that it currently has. Hence, after the interaction $(v_1,v_2)$, the energies of these nodes are updated to $E_1=468.75$ and $E_2 = 131.25$. After that, when $v_1$ and $v_6$ interact, since $v_6$ is the root of the tree and it does not have any target, $v_1$ can just transfer the extra $206.25$ units to $v_6$ in order to reach its target of $265.5$. 
\hfill
\qed
\end{example}

Next, we show that for $\beta=0$ and $k$-ary tree networks, {\bf $k$-depth-target} converges to relaxed energy distribution. For the special case binary trees ($k=2$), the protocol actually converges to an exact up to the root energy distribution, which is also relaxed: the root has energy that is at least twice the energy of each of its children, while any other node has exactly twice the energy of each of its children. 

\begin{thm}
For $\beta=0$ and $k$-ary tree networks, the protocol {\bf $k$-depth-target} converges to relaxed energy distribution. For $k=2$, the energy distribution is also exact up to the root.
\end{thm}

\begin{proof}
Let $\tau \geq 1$ be the step of time after which all nodes correctly estimate their own depths and the height of the tree. Then, for every $t > \tau$ and non-root node $v$, we have that $d_v(t) = d_v(\tau)=d_v$, $h_v(t) = h_v(\tau)=h$, and $\zeta_v(t) = \zeta_v(\tau)=\zeta_v$. During every interaction that a non-root node $v$ participates in, it requests for energy when $E_v(t) < \zeta_v$ and gives away energy when $E_v(t) > \zeta_v$; the root node $r$ gives away and receives energy based on the demands of the other nodes with which it interacts. Therefore, in the first case the energy of a non-root node can only increase until it reaches the target, while in the second case it can only decrease until it again reaches the target. Hence, the protocol converges to a distribution where each non-root node has its target energy, and the root stores the remaining energy of the network. This distribution is exact up to the root. To see this, consider a parent-child pair $(p,c)$ for which it holds that $d_c = d_p + 1$. Then, by \eqref{eq:zeta}, we have
\begin{align*}
\frac{\zeta_p}{\zeta_c} = \left. \left( \frac{E_\total}{k^{d_p}(h + 1)} \right) \middle/ \left( \frac{E_\total}{k^{d_p+1}(h + 1)} \right) \right. = k.
\end{align*}
Hence, the energy distribution is exact up to the root for binary trees ($k=2$) and relaxed for non-root nodes for any other kind of $k$-ary trees ($k \geq 3$). 

To complete the proof and show that the energy distribution is also relaxed, we need to show that the root node also has at least twice the energy of its children. Let $n_d$ be the number of nodes at depth $d$. Since the tree is $k$-ary, we have that $n^d \leq k^d$. Then, the energy of the root node $r$ is 
\begin{align*}
E_r &= E_\total - \sum_{v\neq r}\zeta_v
= E_\total -  \sum_{d=1}^{h} n_d \frac{E_\total}{k^d(h+1)} \\
&= \bigg( h + 1 - \sum_{d=1}^{h}\frac{n_d}{k^d} \bigg) \frac{E_\total}{h+1} 
\geq \frac{E_\total}{h+1}.
\end{align*}
Since any child node $v$ of the root has $\zeta_v = \frac{E_\total}{k(h+1)}$, we obtain that
$E_r\geq k \cdot  \zeta_v \geq 2 \cdot \zeta_v$,
and the theorem follows.
\end{proof}

In addition to the above theorem, we expect that for most instances with binary tree networks, the protocol {\bf $2$-depth-target} will converge to an exact up to the root energy distribution, even in the lossy case: each node will aim to store its target energy, and any energy loss will be accumulated in the energy of the root. However, this means that the distribution will no longer be relaxed since the root might end up having less than twice the energy of its children, for any $k$-ary tree ($k \geq 2$). Of course, converging to an exact up to the root distribution cannot be guaranteed in general, especially in cases where the energy loss is so much that there is not enough energy for non-root nodes to satisfy their targets. 

Further, observe that when $\beta=0$ and the network is a complete binary tree (the root and all internal nodes have exactly two children), the target energy of any node will be exactly equal to its ideal energy. Hence, {\bf $2$-depth-target} converges to the unique exact distribution in this case. However, if the tree structure is incomplete (with the worst case being a line), then the target energy of any non-root node will be far away from ideal. As we will see in our simulations for binary trees, the structure of the network plays a huge factor in the quality of the final distribution when compared to the distribution produced by {\bf ideal-target}.

\section{Simulations}\label{sec:experiments}

In this section, we evaluate the energy redistribution protocols defined in Section~\ref{sec:energy-protocols}, via indicative simulations implemented in Matlab 2018a. 

\subsection{Simulation Setup}
We study both the lossless case where $\beta=0$, and the lossy case where $\beta$ is a random variable following the Normal Distribution $N(0.2, 0.05)$. We investigate two scenarios regarding the initial energy of the nodes. The total initial network energy is analogous to the number $n \in \{10, 30, 50\}$ of nodes in the network (in particular, $n \cdot 10^3$), and it can be split amongst them either {\em uniformly} leading to all nodes having the same initial energy, or {\em randomly} leading to the nodes having possibly different initial energy. The second scenario (of unequal initial energy per node) is more realistic due to the different characteristic and needs of the various portable devices there exist (for example, smartphones). 

For simplicity and in order to have meaningful results about the protocol {\bf $k$-depth-target}, we only focus on binary tree networks (with $k=2$); in what follows we refer to {\bf $2$-depth-target} as {\bf depth-target}. 
For statistical smoothness we have repeated each simulation $100$ times. The statistical analysis of our findings (average, median, lower and upper quartiles, outliers) demonstrate very high concentration around the mean and, thus, in the following we depict only the average values of our simulation results.

\subsection{Metrics}
We use two metrics to measure the performance of our protocols. The first one is called {\em distribution distance} and is designed to show how fast a protocol converges to a relaxed energy distribution if this is possible, while the second one is called {\em energy distance} and is a measure of the quality of the final energy distribution compared to the ideal one (as computed by {\bf ideal-target} given the true total initial network energy; see Section~\ref{sec:ideal-target} for the definition of the protocol).

For any step of time $t \geq 0$, the distribution distance is defined as 
\begin{align}\label{eq:DD}
\text{DD}(t) = \sum_u \sum_{v: u \haschild v} \bigg( 2E_v(t) - E_u(t) \bigg) \one\{E_u(t) < 2E_v(t)\}.
\end{align}
Essentially, the distribution distance counts the total energy that must be redistributed at time $t$ in order to achieve a relaxed energy distribution. The {\em convergence time} of {\bf $\lambda$-exchange} and {\bf $\kappa$-transfer} is the first time step $\tau$ for which $\text{DD}(\tau) = 0$, when a relaxed distribution has been reached. For the targeted protocols {\bf ideal-target} and {\bf depth-target}, the convergence time is the first time step $\tau$ after which the distribution distance remains constant: $\text{DD}(\tau) = \text{DD}(t)$ for every $t > \tau$. The value of $\text{DD}(\tau)$ will be $0$ in the lossless case since both protocols achieve a relaxed distribution then. In the lossy case however, $\text{DD}(\tau)$ can be positive, since the energy that is lost may prevent the protocols to achieve a relaxed distribution, especially {\bf ideal-target} which sets the targets so as to achieve an exact distribution based on the initial total network energy.

The energy distance of the final distribution from the ideal one is defined as
\begin{align}\label{eq:ED}
\text{ED} = \frac{1}{2} \sum_v |E_v(\tau) - \gamma_v(T,E_\total)|,
\end{align}
where $\gamma_v(T,E_\total)$ is the target energy of $v$ as set by {\bf ideal-target}, i.e., it is the ideal energy based on the underlying tree structure $T$ of the network and the total initial network energy. Essentially, the energy distance counts the total energy that has been misplaced in the final distribution. This metric also allows us to compare our protocols even when there is energy loss; in this case, to compute the ideal energy of each node, we consider the total energy in the network at the time when the network has been formed. Since even a ``perfect'' protocol like {\bf ideal-target} will inevitably lose some energy in the lossy case, we also measure the energy that has been lost until the convergence to a steady energy distribution.

\subsection{Fine-tuning of {\bf $\lambda$-exchange} and {\bf $\kappa$-transfer}}\label{sec:finetuning}
To optimize the performance of {\bf $\lambda$-exchange} and {\bf $\kappa$-transfer}, 
we study the cases where $\lambda \in \{2, 3, 4, 5, 6\}$ and $\kappa \in \{0.3, 0.4, 0.5, 0.6, 0.7\}$.

\begin{figure}[t!]
\centering
\begin{subfigure}{0.49\textwidth}
	\includegraphics[scale=0.4]{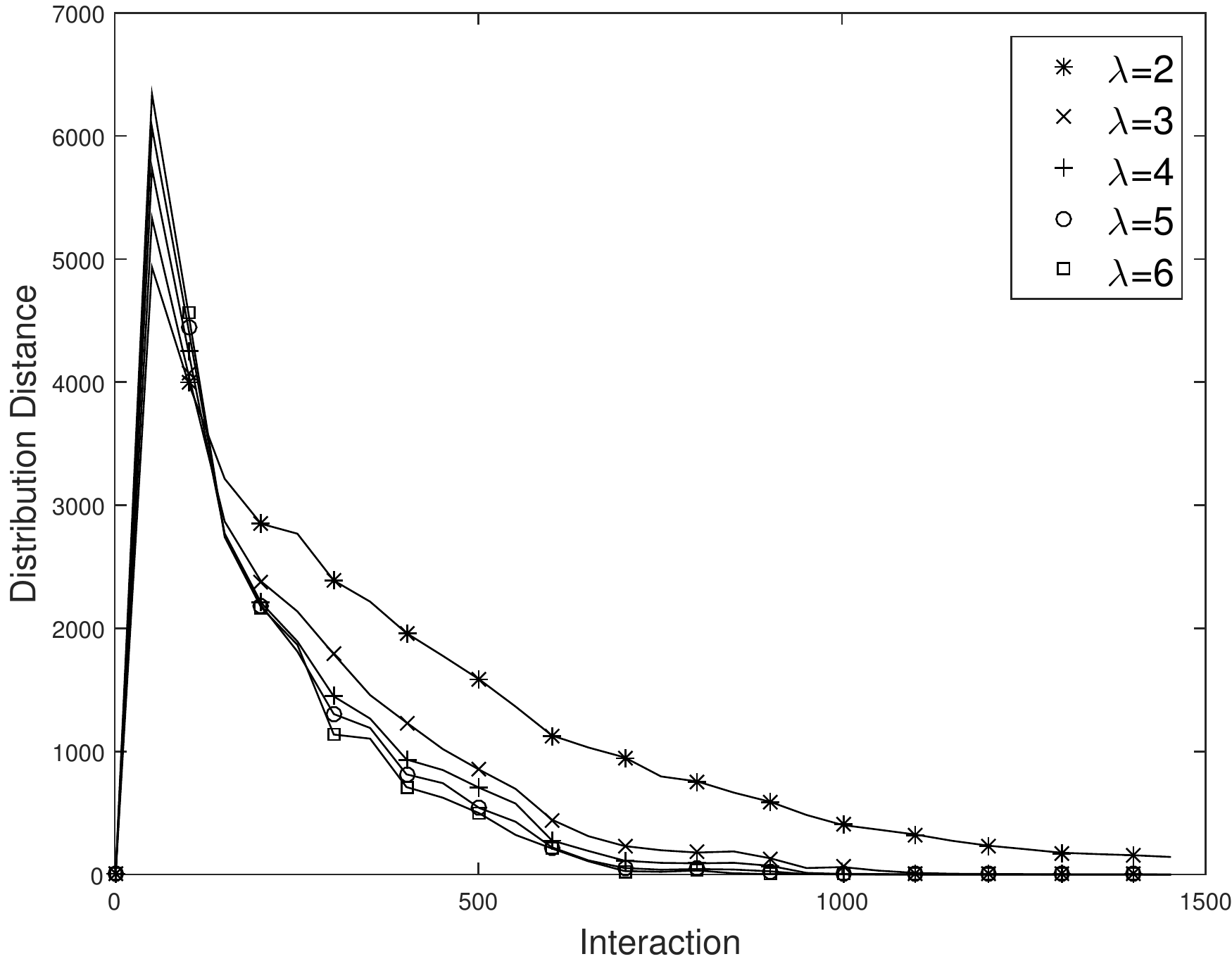}
	\caption{{\bf $\lambda$-exchange}}
   \label{subfig:exchange}
\end{subfigure}
\begin{subfigure}{0.49\textwidth}
	\includegraphics[scale=0.4]{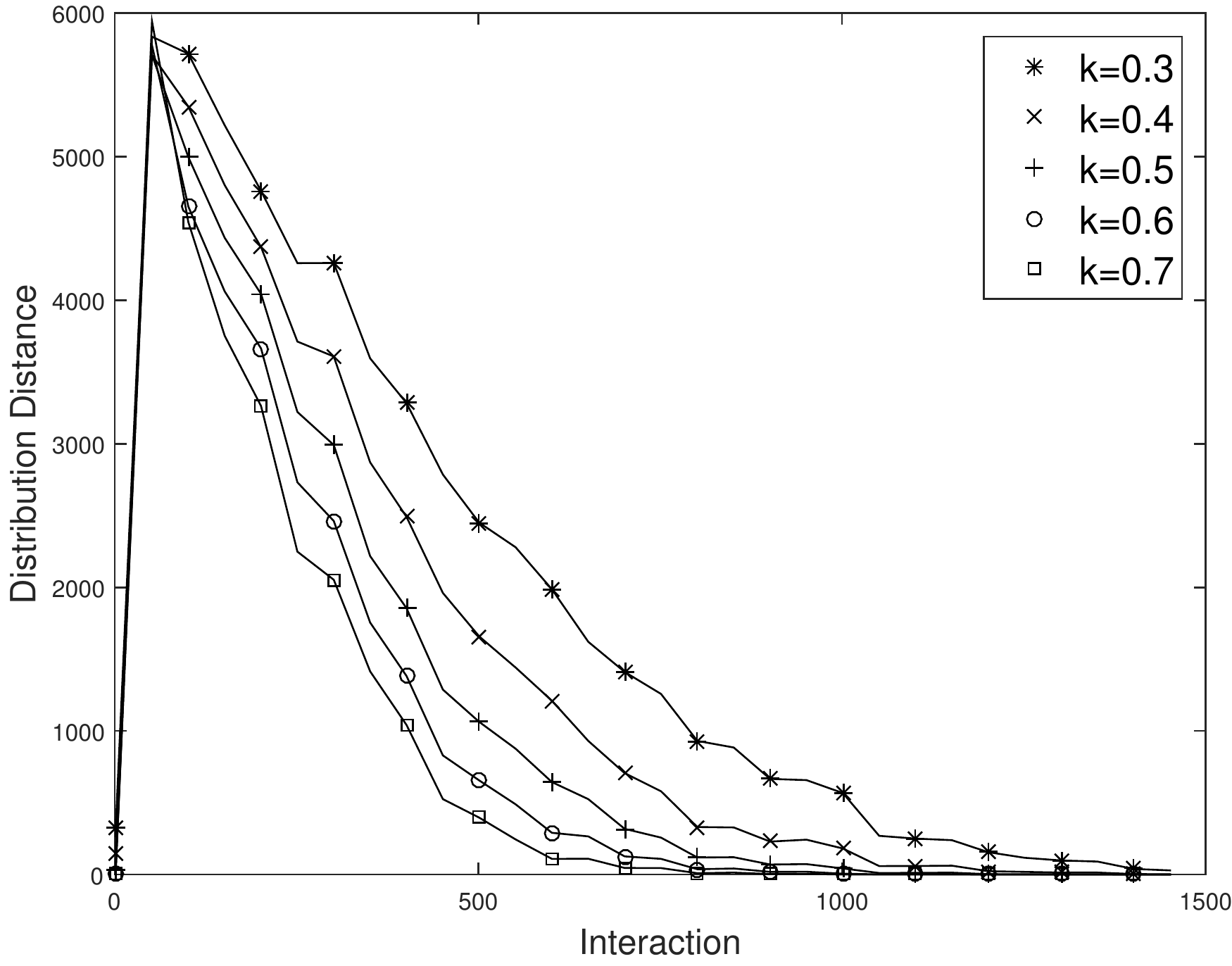}
	\caption{{\bf $\kappa$-transfer}}
   \label{subfig:transfer}
\end{subfigure}
\caption{Distribution distance of {\bf $\lambda$-exchange} and {\bf $\kappa$-exchange} for various values of $\lambda \geq 2$ and $\kappa \in [0,1]$. It is easy to observe that as $\lambda$ and $\kappa$ increase, the distance decreases faster.}
\label{fig:finetune}
\end{figure}

\begin{table}[t!]
\centering
\begin{tabular}{|c|c|c|c|c|}
\noalign{\hrule height 1pt}\hline
\multicolumn{5}{|c|} {{\bf $\lambda$-exchange}} \\ 
\noalign{\hrule height 1pt}\hline
$\lambda=2$ & $\lambda=3$ & $\lambda=4$ & $\lambda=5$ & $\lambda=6$  \\
\hline
5,98\% & 12,76\% &  18,03\% & 21,03\% & 23,92\% \\
\noalign{\hrule height 1pt}\hline
\end{tabular}
\ \ \ \ \
\begin{tabular}{|c|c|c|c|c|}
\noalign{\hrule height 1pt}\hline
\multicolumn{5}{|c|} {{\bf $\kappa$-transfer}} \\ 
\noalign{\hrule height 1pt}\hline
$\kappa=0.3$ & $\kappa=0.4$ & $\kappa=0.5$ & $\kappa=0.6$ & $\kappa=0.7$  \\
\hline
57,80\% & 59,44\% &  59,65\% & 59,67\% & 59,67\% \\
\noalign{\hrule height 1pt}\hline
\end{tabular}
\caption{Energy distance of {\bf $\lambda$-exchange} and {\bf $\kappa$-transfer} for various values of $\lambda \geq 2$ and $\kappa \in [0,1]$. The depicted percentages are the result of diving the energy distance by the total network energy. It is easy to observe that for {\bf $\lambda$-exchange} as $\lambda$ increases, the distance increases. In contrast, for {\bf $\kappa$-transfer} all values of $\kappa$ exhibit similar high energy distance.}
\label{tab:finetune}
\end{table}

Figure~\ref{fig:finetune} depicts the performance of the two protocols for these values of the parameters with respect to the distribution distance metric, for $n=10$ nodes in the lossless case and different initial energy supplies; we have experimented with many more different settings, but the conclusions are similar. 
One can easily observe that as the values of $\lambda$ and $\kappa$ increase, the distribution distance and, consequently, the convergence time decreases. 
Specifically, for {\bf $\lambda$-exchange} there is a huge improvement from $\lambda=2$ to $\lambda=3$, but then the improvement is only minor as we consider higher values of $\lambda$; this seems to be a consequence of the energy distribution definition, where we aim for nodes to have twice the energy of each of their children, and hence there is a significant improvement when we move from $\lambda=2$ to $\lambda=3$. 
In contrast, for {\bf $\kappa$-transfer} the improvement in convergence time is steady as the value of $\kappa$ increases, which is also expected since the energy transfers from child to parent nodes increase with  rate analogous to $\kappa$. 

Table~\ref{tab:finetune} contains the values of the energy distance induced by the two protocols for the values of the parameters considered. We can observe that higher values of $\lambda$ incur higher and higher energy distance, while the energy distance remains pretty much at the same level for all values of $\kappa$.

Given these observations, we conclude that the values $\lambda=2$ and $\lambda=3$ nicely balance the distribution distance and the energy distance for {\bf $\lambda$-exchange}. Hence, we choose $\lambda=2$, but also implement another {\em randomized} protocol, which we call {\bf rand-exchange}. According to this protocol, $\lambda$ varies from interaction to interaction, as a random variable following a uniform distribution taking values in the interval $[2,3]$. For {\bf $\kappa$-transfer} we choose $k=0.5$ as the middle ground.

\subsection{Comparison of the protocols}
Figures~\ref{fig:distribution_distance_lossless} and~\ref{fig:distribution_distance_lossy} depict the performance of our energy redistribution protocols in terms of the distribution distance metric in the lossless and lossy case, respectively. The corresponding performance values in terms of the energy distance metric are given in Tables~\ref{tab:EnergyDistance-lossless} and \ref{tab:EnergyDistance-lossy}, while Table~\ref{tab:Energy-loss} contains the average percentages of energy that has been lost until the protocols converge to a steady energy distribution. See Figure~\ref{fig:tree-heatmap} for an example of a tree structure created via our simulations and the corresponding final distributions achieved by the protocols.

First, we can easily observe that as the number of nodes increases, the protocols require more time to converge (see Figures~\ref{fig:distribution_distance_lossless} and~\ref{fig:distribution_distance_lossy}), the quality of the final distribution decreases (see Tables~\ref{tab:EnergyDistance-lossless} and \ref{tab:EnergyDistance-lossy}), and more energy is lost (see Table~\ref{tab:Energy-loss}). Second, whether the nodes all have equal or different energy does not seem to affect the performance of the protocols (for instance, compare Figures \ref{fig:lossless_uniform_10}, \ref{fig:lossless_uniform_30} and \ref{fig:lossless_uniform_50} to \ref{fig:lossless_non_uniform_10}, \ref{fig:lossless_non_uniform_30} and \ref{fig:lossless_non_uniform_50}). Finally, in the lossy case the protocols converge faster than in the lossless case (for example, compare Figures \ref{fig:lossless_uniform_10}--\ref{fig:lossless_uniform_50} to Figures~\ref{fig:loss_uniform_10}--\ref{fig:loss_uniform_50}). This is expected since in the lossy case the total network energy is decreasing in time, and therefore there is less energy that can be transferred from node to node.

In all cases, {\bf $2$-exchange} seems to be the slowest protocol, which is natural since it is not a targeted protocol and may require many exchanges where the energy may go either upwards (to parent nodes) or downwards (to child nodes), depending on the needs of the interacting nodes. In contrast, for the {\bf $0.5$-transfer} protocol the energy can only go towards the root of the tree, since at every parent-child interaction where the energy equilibrium is not satisfied, the child transfers half of its energy to the parent. The protocol {\bf rand-exchange} is of course faster than {\bf $2$-exchange}, since it also randomly considers higher values of $\lambda$; see the previous subsection for a comparison of different values of $\lambda$. On the other hand, {\bf depth-target} seems to be the fastest protocol, outperforming even {\bf ideal-target}. This is due to its definition, where the root plays the role of an auxiliary node helping the other nodes to reach their target, as opposed to the case of {\bf ideal-target} where the root also has a target that must be reached. 

Even though {\bf $2$-exchange} needs more time to converge to a stable energy distribution, as we can see in Tables~\ref{tab:EnergyDistance-lossless} and \ref{tab:EnergyDistance-lossy}, it outperforms all protocols in terms of energy distance, except {\bf ideal-target} in the lossless case, which of course has zero energy distance by its definition. The low energy distance that {\bf $2$-exchange} achieves in the lossless case is a due to the energy equilibrium condition that it exploits, which is tailor-made so that when a parent-child pair of nodes interact, the parent ends up with exactly twice the energy of the child. Consequently, we should expect that the final energy distribution will be quite close to the ideal one, especially when no energy is lost.
Even though the energy distance of {\bf $2$-exchange} is considerably higher in the lossy case, it is still lower than that of the other protocols, including {\bf ideal-target}, which is quite surprising if one also accounts the fact that the energy loss is much smaller for {\bf ideal-target} and {\bf depth-target}; see Table~\ref{tab:Energy-loss}. 
Of course, the {\bf rand-exchange} has a comparable performance in this perspective since it is similarly defined.

The protocols {\bf $0.5$-transfer} and {\bf depth-target} have high energy distance in both the lossless and the lossy case, while {\bf ideal-target} also has high energy distance in the lossy case. This is naturally expected for {\bf $0.5$-transfer} since it only transfers energy towards the root and therefore creates an extremely unbalanced energy distribution, which might end up being relaxed, but it is much different than the ideal one.  
On the other hand, the reason why {\bf depth-target} has such a high energy distance might be due to the fact that the target energy of each node depends on the network structure, and if the structure is unbalanced, then almost all of the network energy ends at the root node, which plays an auxiliary role. For {\bf ideal-target} this is not a problem since the root also has a particular target energy, but in the lossy case the energy that is lost creates the following problem: for an energy transfer between two nodes, one of them must need to give away excess energy and the other must need to receive energy; since energy is lost during the exchanges, all nodes might end up needing either to give away energy or to receive energy.

\begin{figure}[t]
\centering
\includegraphics[scale=0.55]{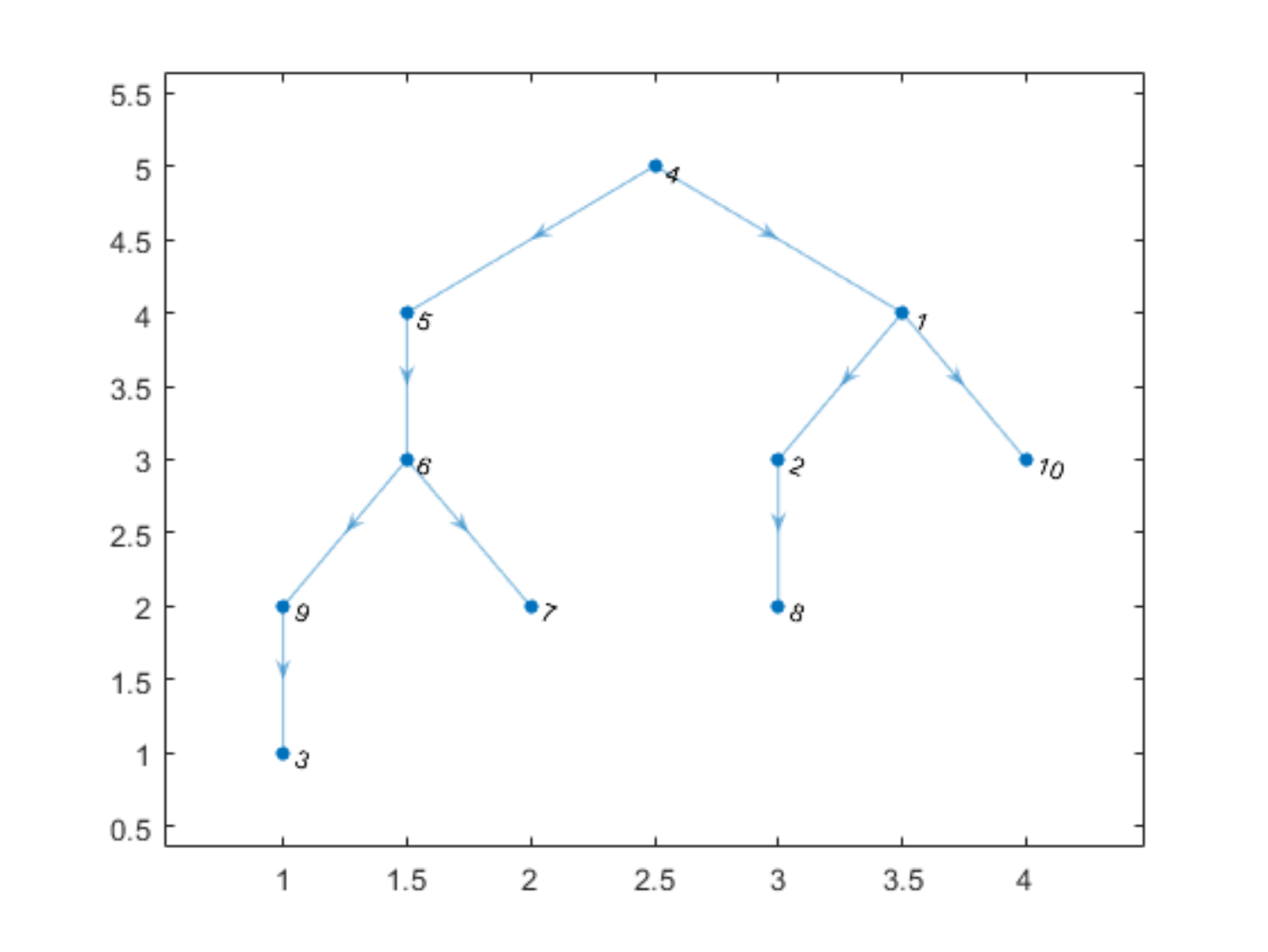}
\includegraphics[scale=0.55]{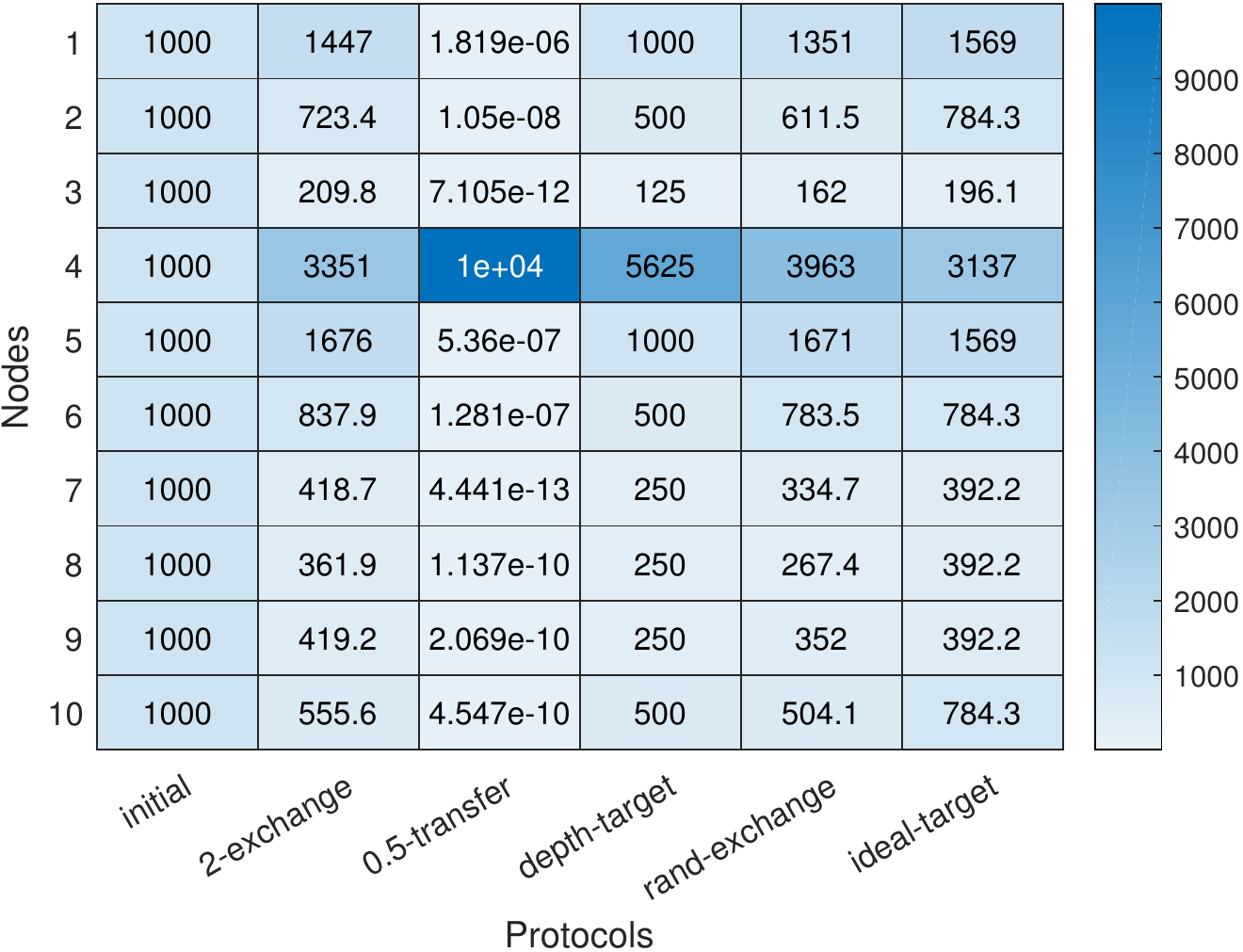}
\caption{An example of a binary tree network for $n=10$ nodes and the corresponding final energy distribution that the protocols compute for the lossless case with equal initial energy for the nodes.}
\label{fig:tree-heatmap}
\end{figure}

Given the above observations, the protocols that best balance all of the trade-offs we have discussed (fast convergence and energy distribution of good quality, with low energy loss until convergence) seem to be {\bf $2$-exchange} and {\bf rand-exchange}, which are also oblivious and do not depend on global knowledge; {\bf depth-target} is of course a good candidate, but it is vulnerable when the tree network ends up being even just a bit unbalanced, and requires knowledge that may not be available a priori, which is a hard constraint when dealing with computationally weak devices with limited power and memory, as we assume in this work.

\begin{figure*}[t]
\centering
\begin{subfigure}{0.49\textwidth}
   \includegraphics[width=0.9\columnwidth]{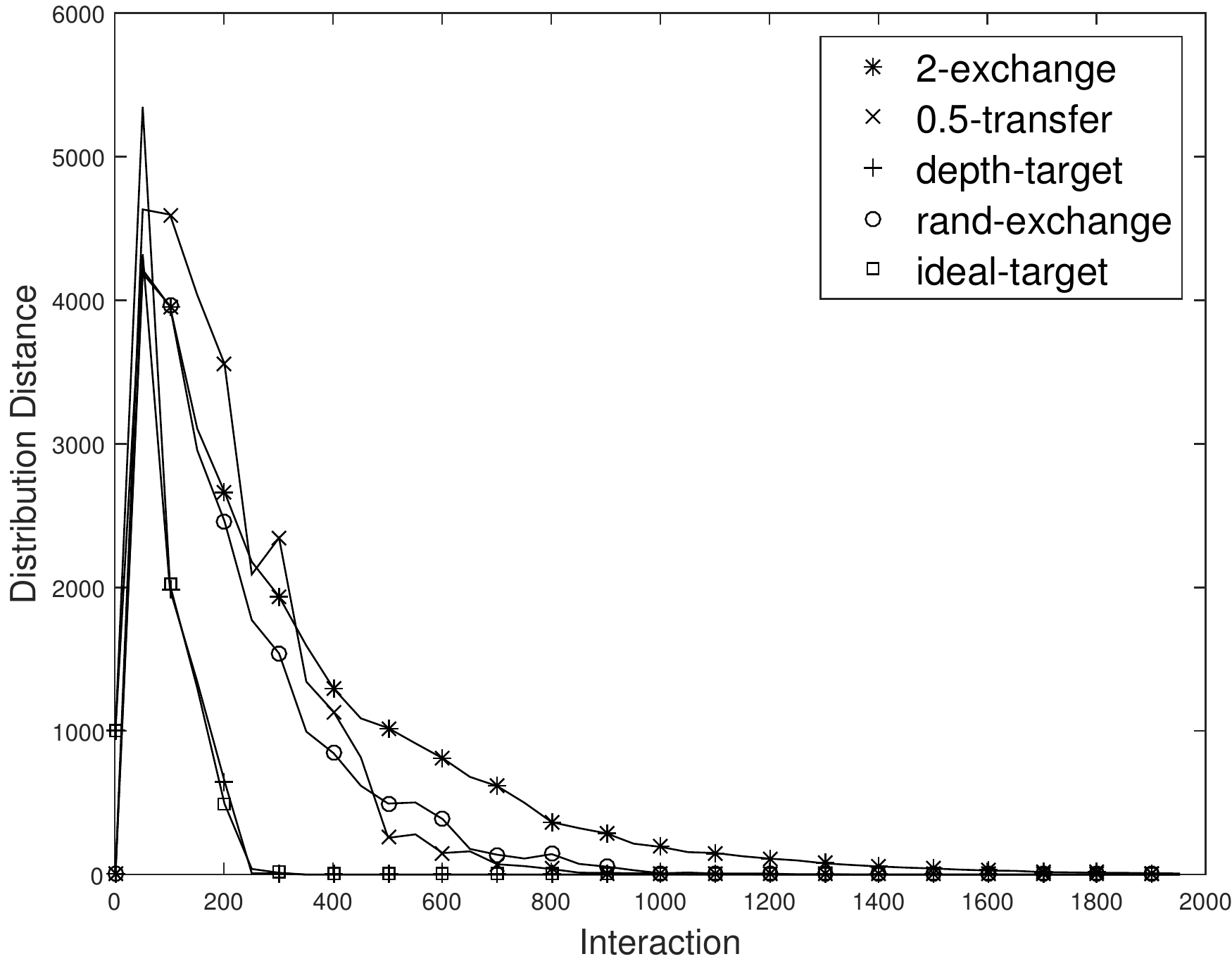}
   \caption{$n=10$}
   \label{fig:lossless_uniform_10}
\end{subfigure}
\begin{subfigure}{0.49\textwidth}
   \includegraphics[width=0.9\columnwidth]{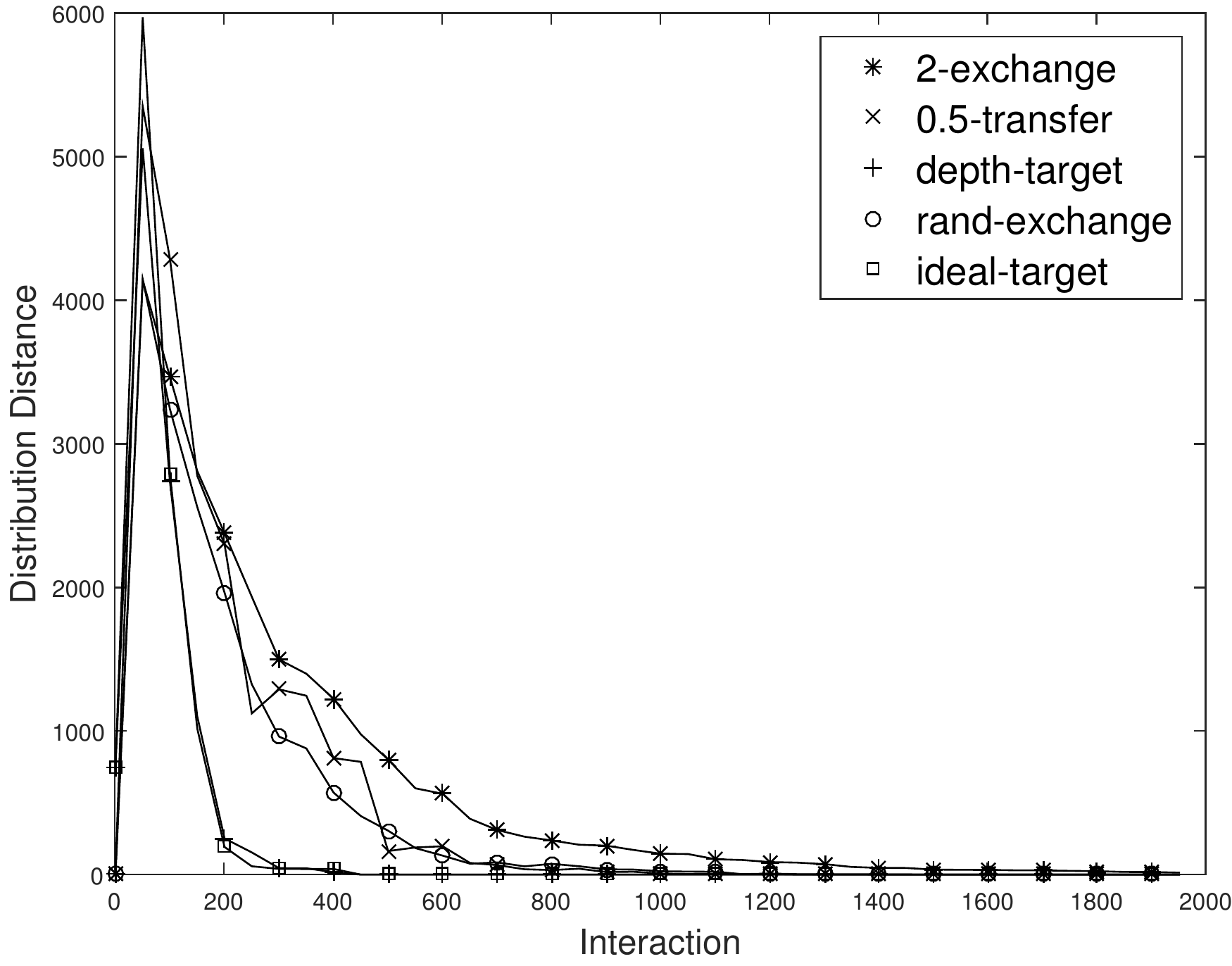}
   \caption{$n=10$}
   \label{fig:lossless_non_uniform_10}
\end{subfigure}

\begin{subfigure}{0.49\textwidth}
   \includegraphics[width=0.9\columnwidth]{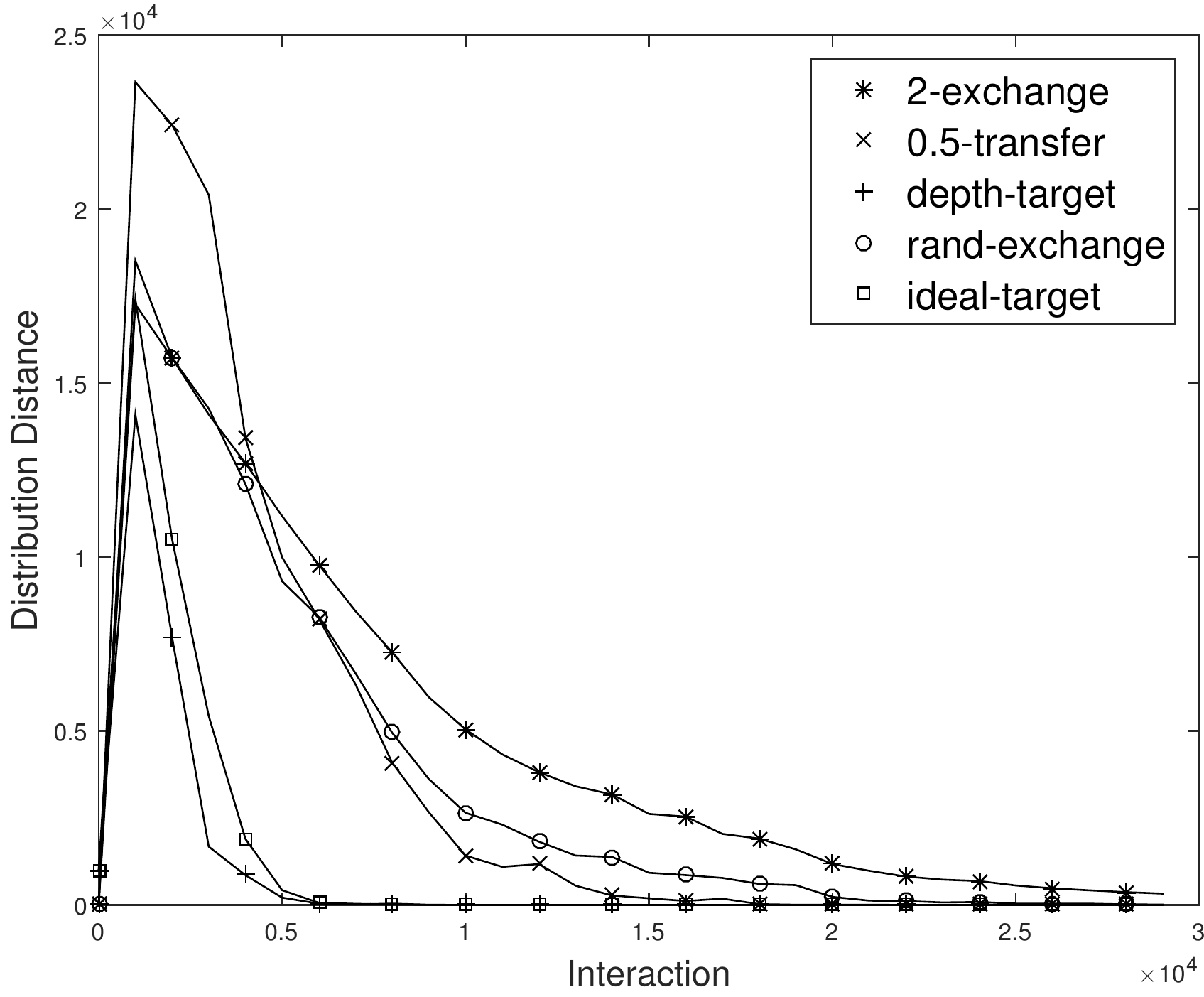}
   \caption{$n=30$}
   \label{fig:lossless_uniform_30}
\end{subfigure}
\begin{subfigure}{0.49\textwidth}
   \includegraphics[width=0.9\columnwidth]{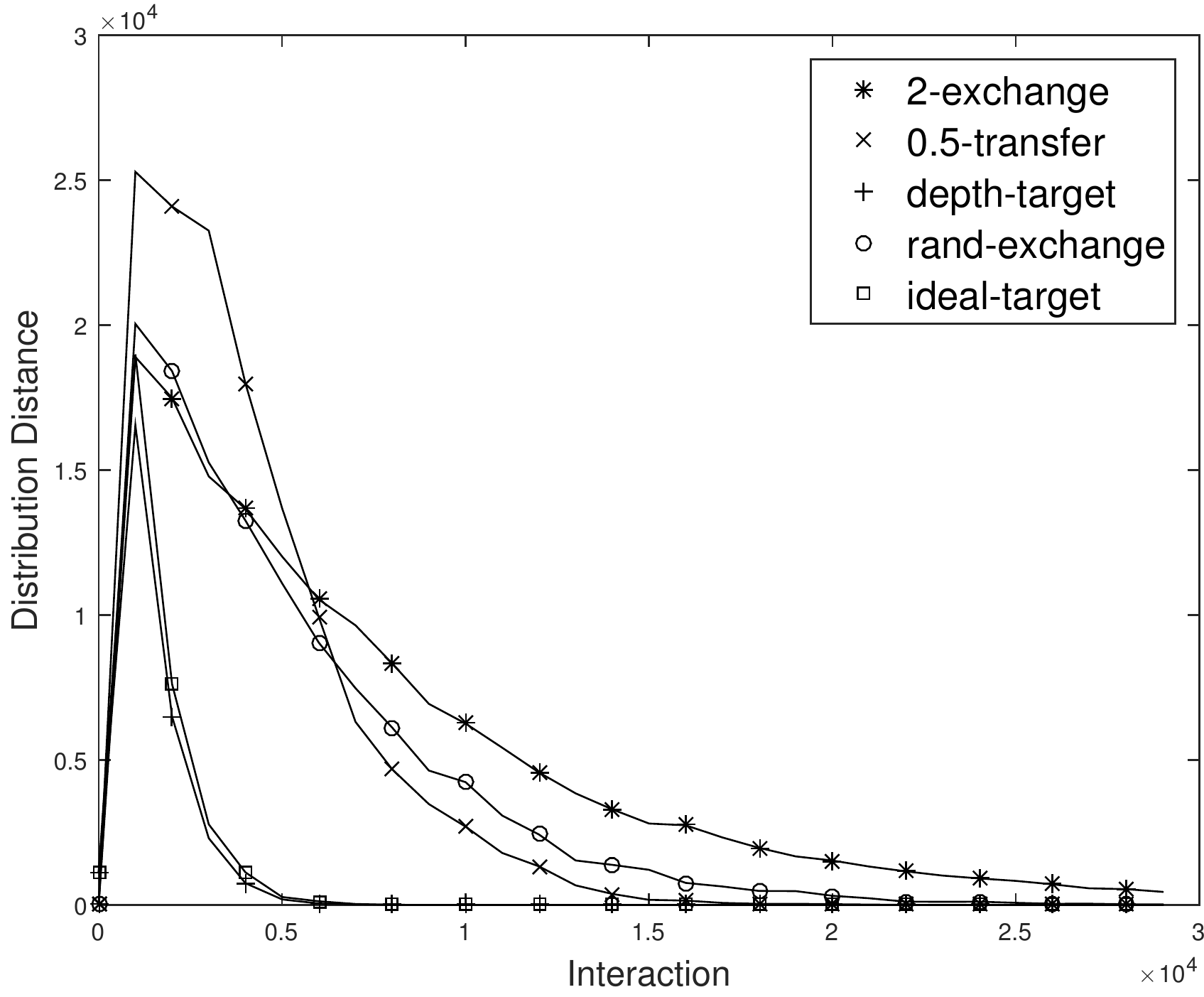}
   \caption{$n=30$}
    \label{fig:lossless_non_uniform_30}
\end{subfigure}

\begin{subfigure}{0.49\textwidth}
   \includegraphics[width=0.9\columnwidth]{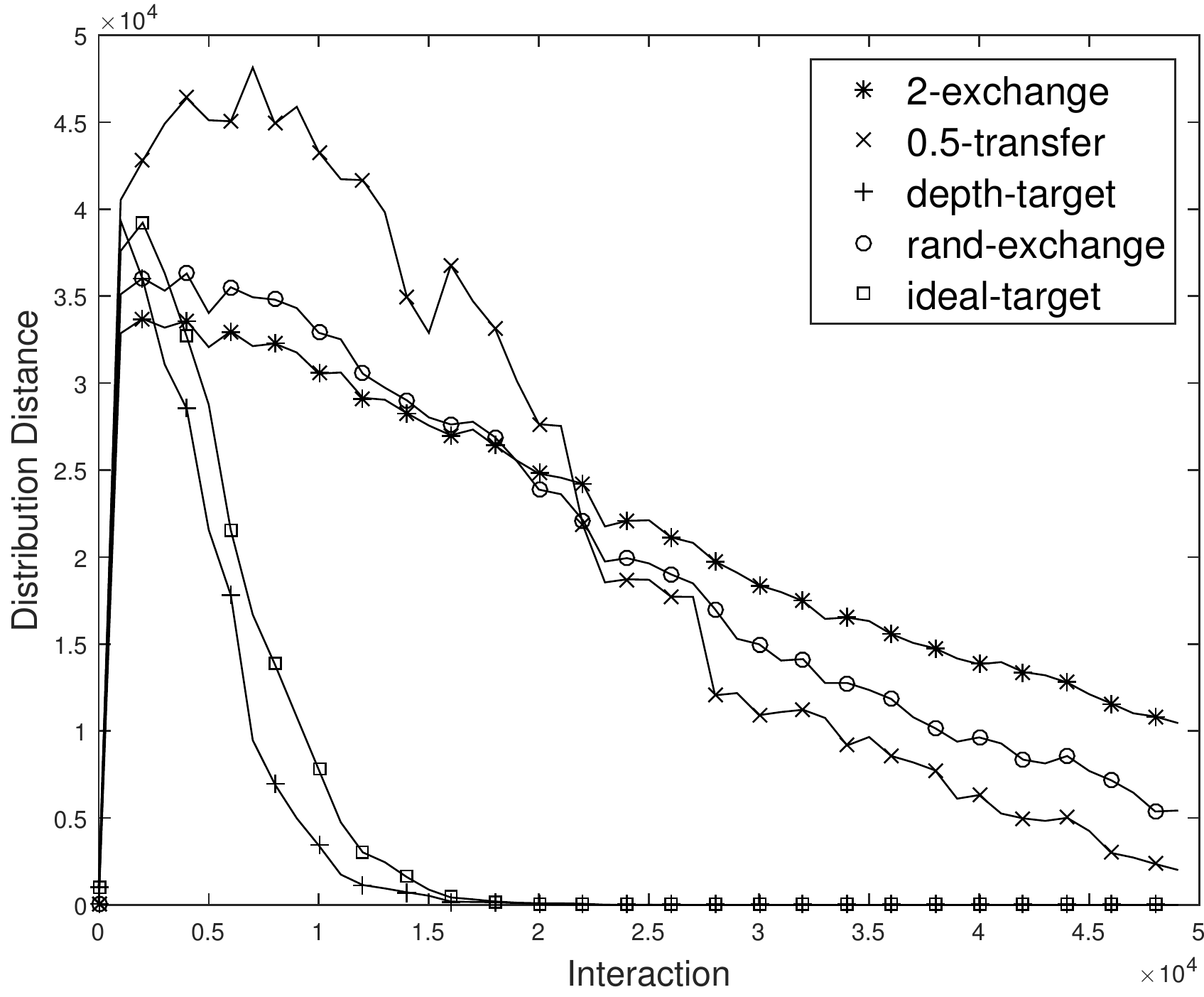}
   \caption{$n=50$}
    \label{fig:lossless_uniform_50}
\end{subfigure}
\begin{subfigure}{0.49\textwidth}
   \includegraphics[width=0.9\columnwidth]{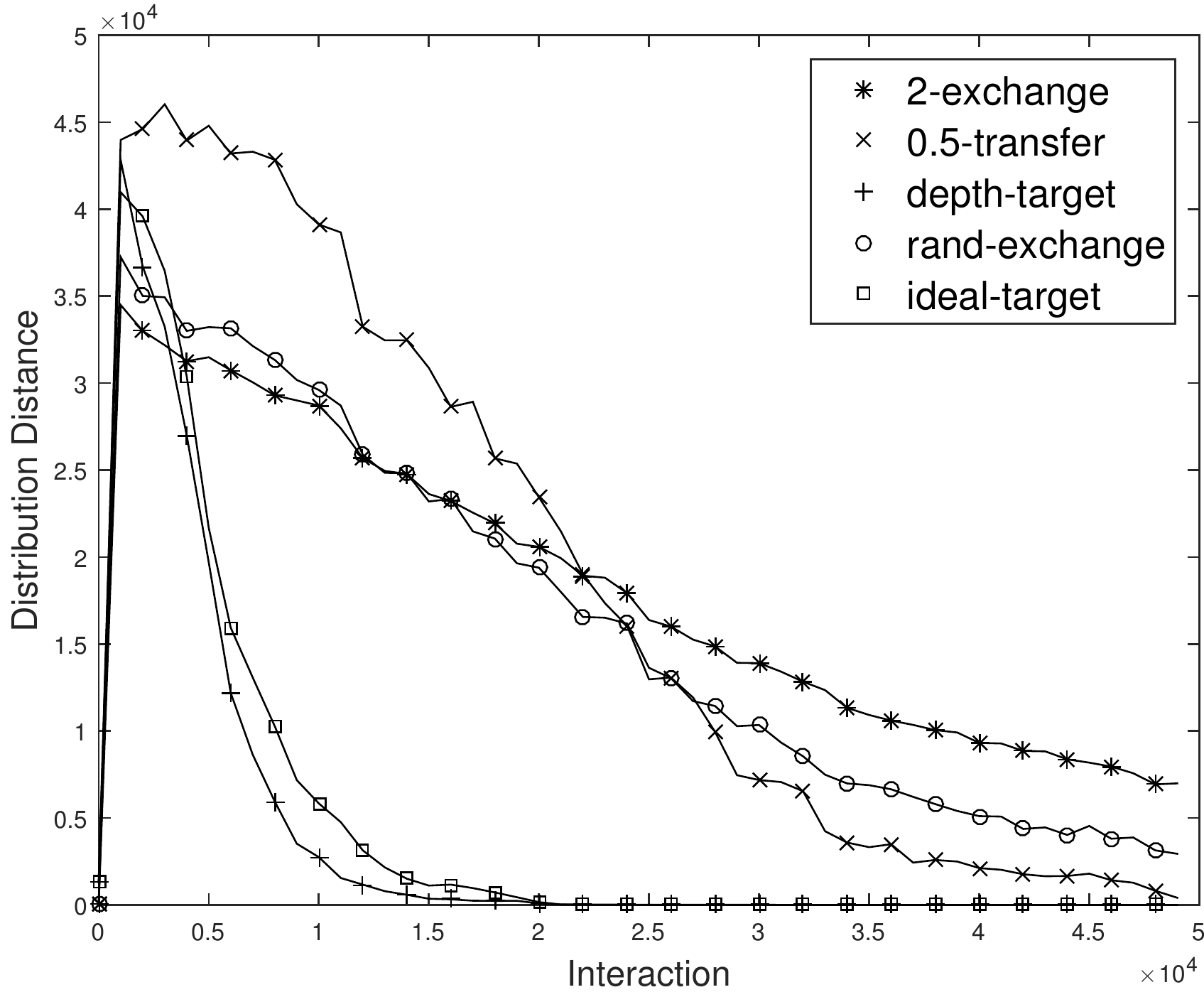}
   \caption{$n=50$}
   \label{fig:lossless_non_uniform_50}
\end{subfigure}
\caption{Distribution distance of the five energy redistribution protocols ({\bf $2$-exchange}, {\bf $0.5$-transfer}, {\bf depth-target}, {\bf rand-exchange} and {\bf ideal-target}) in the {\em lossless} case for various number of nodes $n\in \{10,30,50\}$. Figures (a), (c) and (e) depict the performance of the protocols in the case where the nodes have the same initial energy, while Figures (b), (d) and (f) depict their performance in the case of different initial energy supplies.}
\label{fig:distribution_distance_lossless}
\end{figure*}

\begin{figure*}[t]
\centering
\begin{subfigure}{0.49\textwidth}
   \includegraphics[width=0.9\columnwidth]{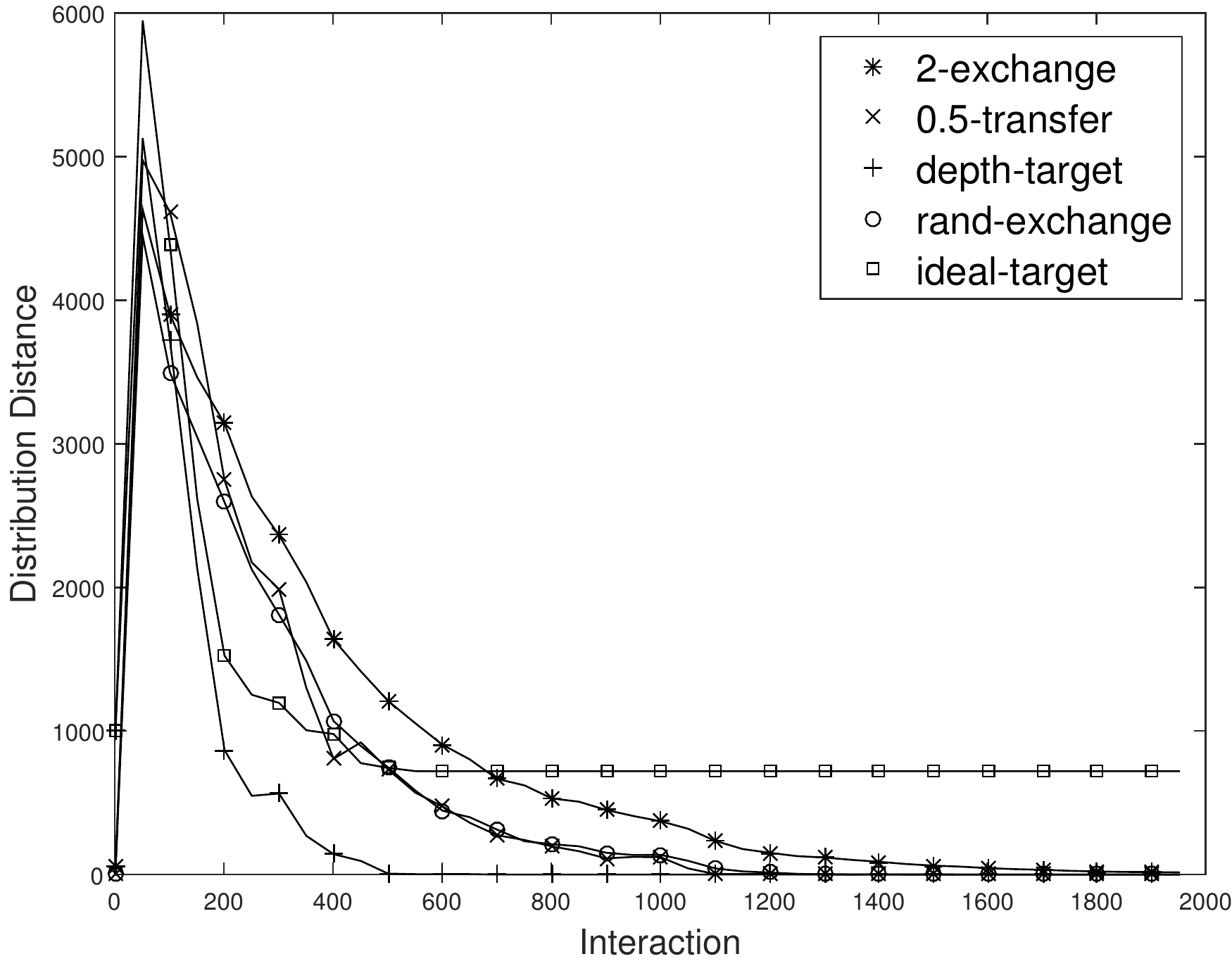}
   \caption{$n=10$}
   \label{fig:loss_uniform_10}
\end{subfigure}
\begin{subfigure}{0.49\textwidth}
   \includegraphics[width=0.9\columnwidth]{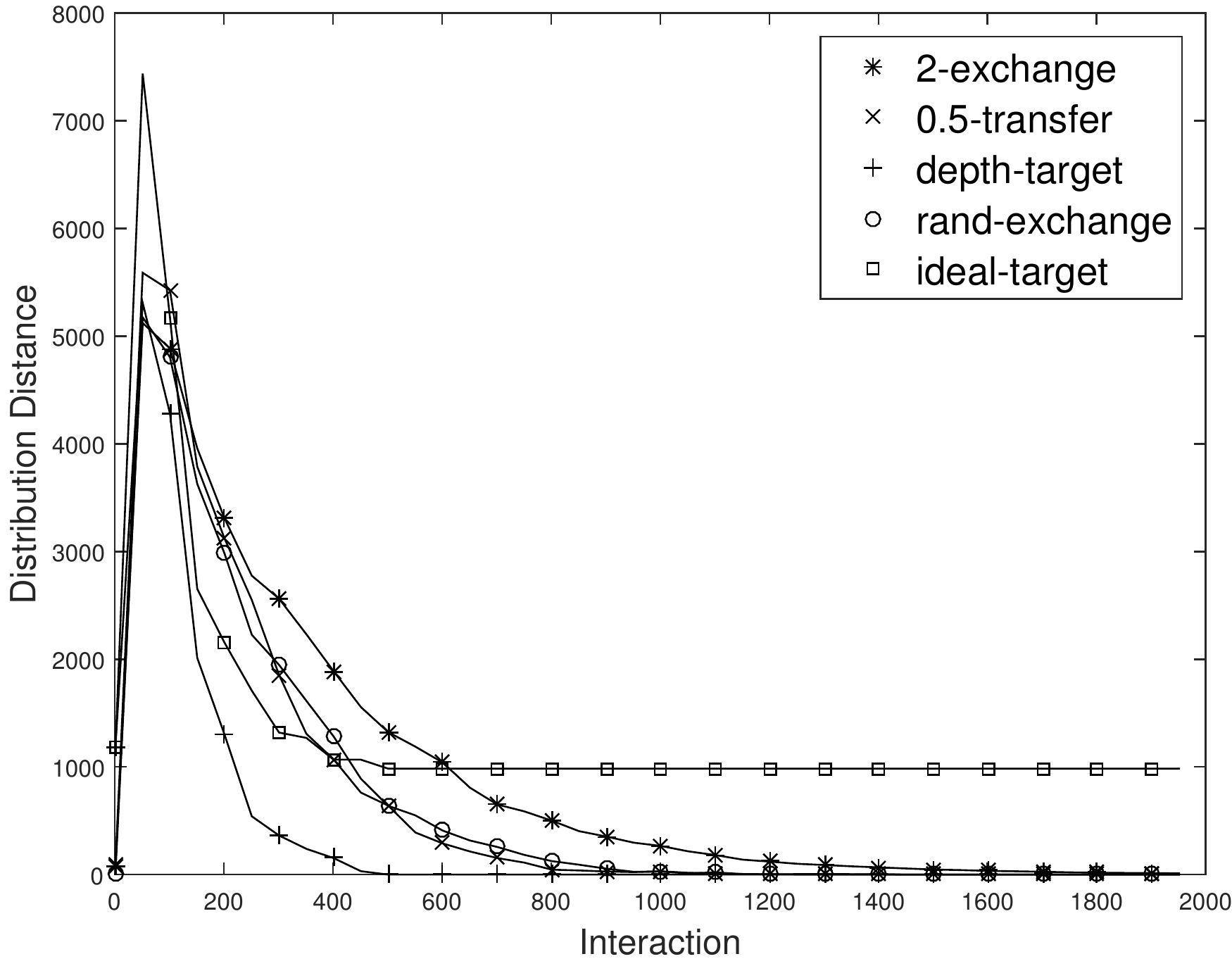}
   \caption{$n=10$}
   \label{fig:loss_non_uniform_10}
\end{subfigure}

\begin{subfigure}{0.49\textwidth}
   \includegraphics[width=0.9\columnwidth]{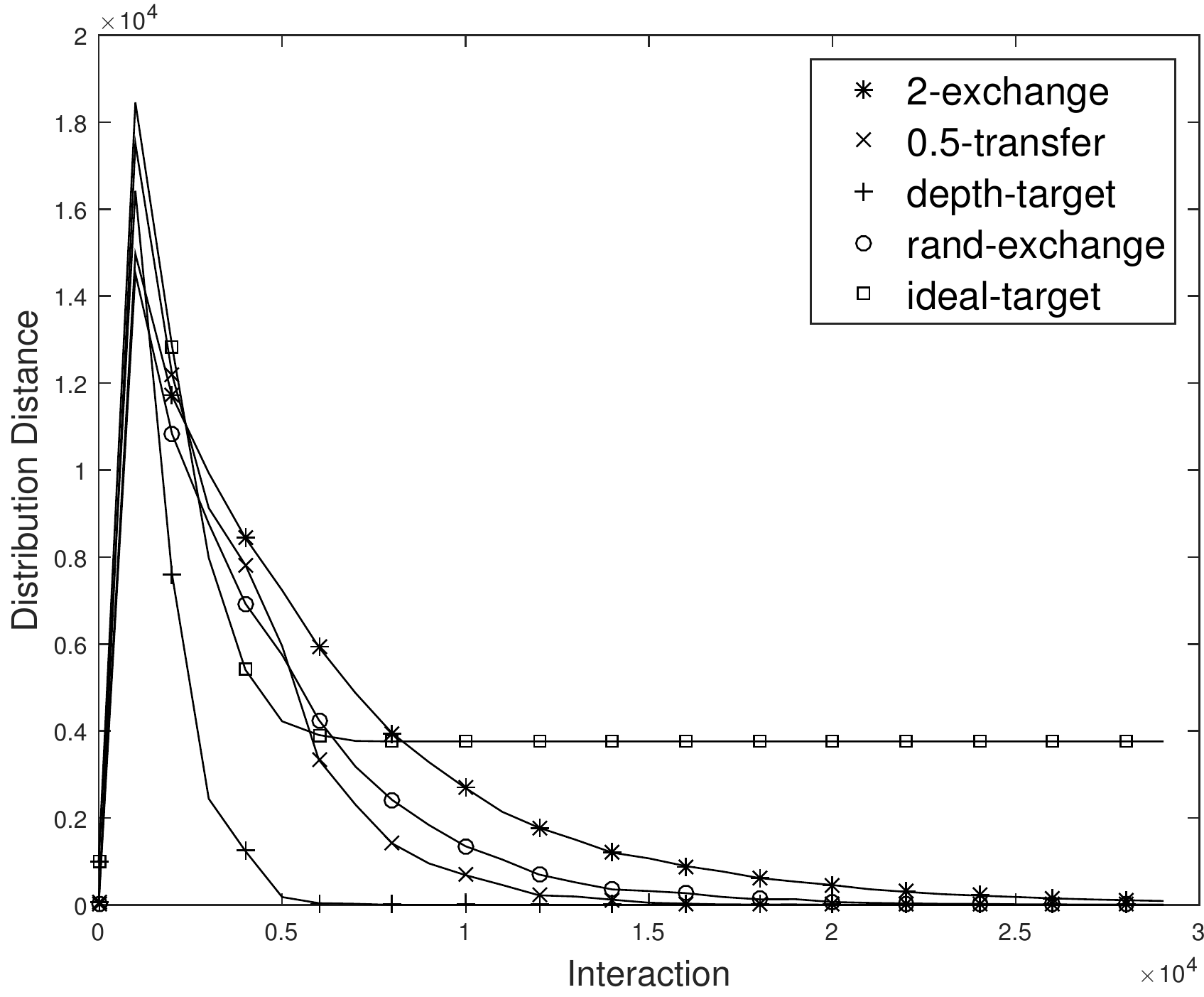}
   \caption{$n=30$}
   \label{fig:loss_uniform_30}
\end{subfigure}
\begin{subfigure}{0.49\textwidth}
   \includegraphics[width=0.9\columnwidth]{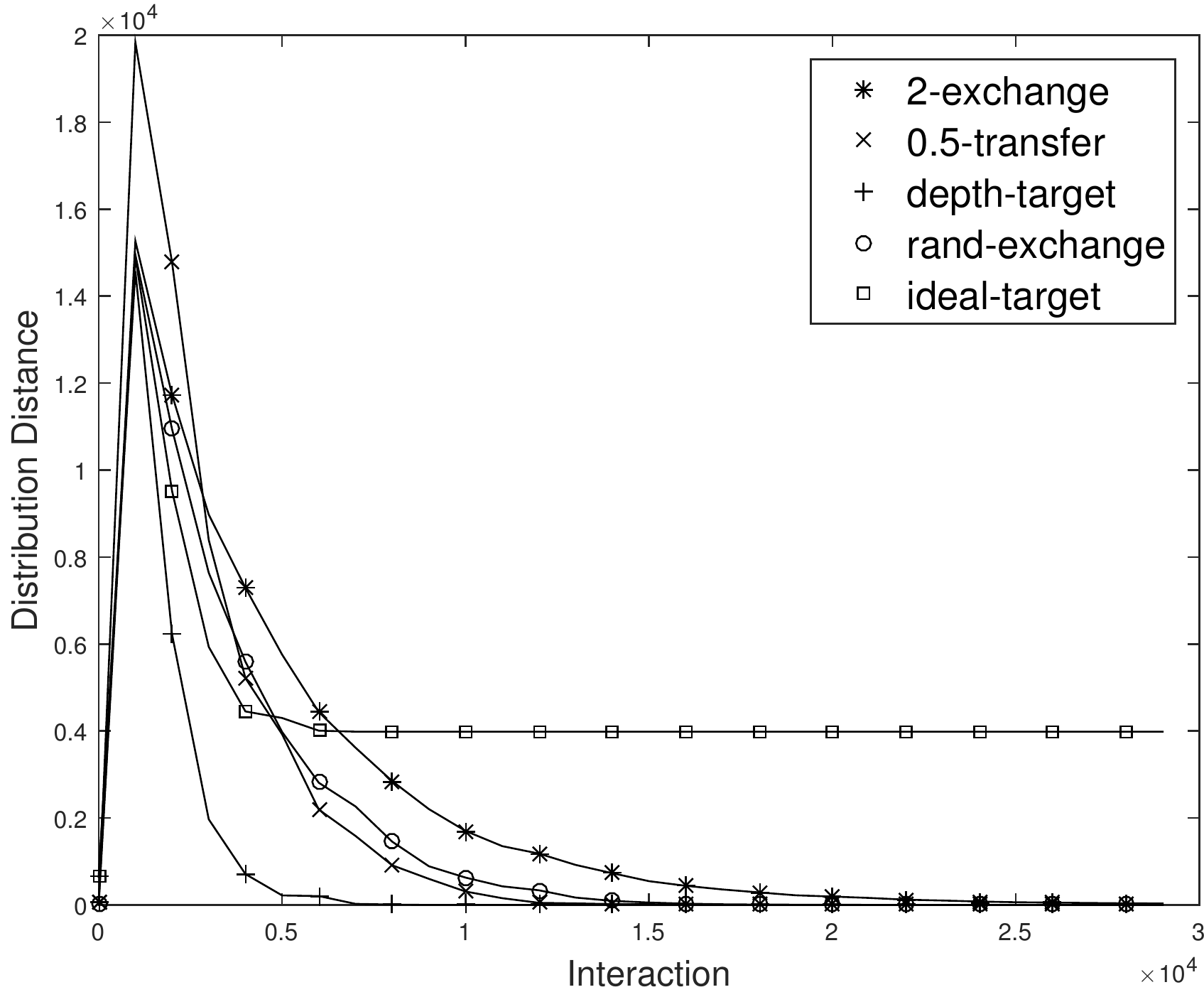}
   \caption{$n=30$}
    \label{fig:loss_non_uniform_30}
\end{subfigure}

\begin{subfigure}{0.49\textwidth}
   \includegraphics[width=0.9\columnwidth]{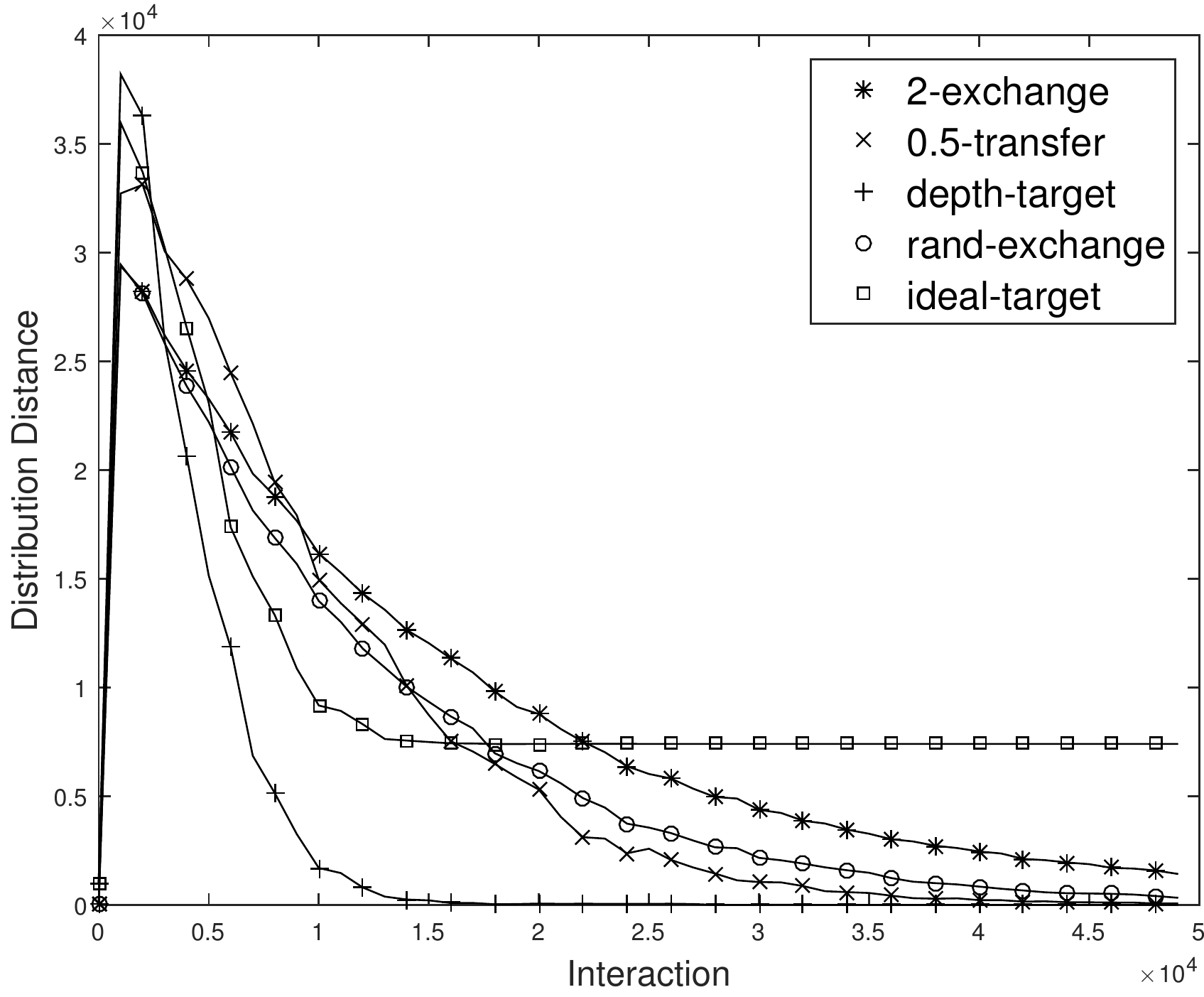}
   \caption{$n=50$}
    \label{fig:loss_uniform_50}
\end{subfigure}
\begin{subfigure}{0.49\textwidth}
   \includegraphics[width=0.9\columnwidth]{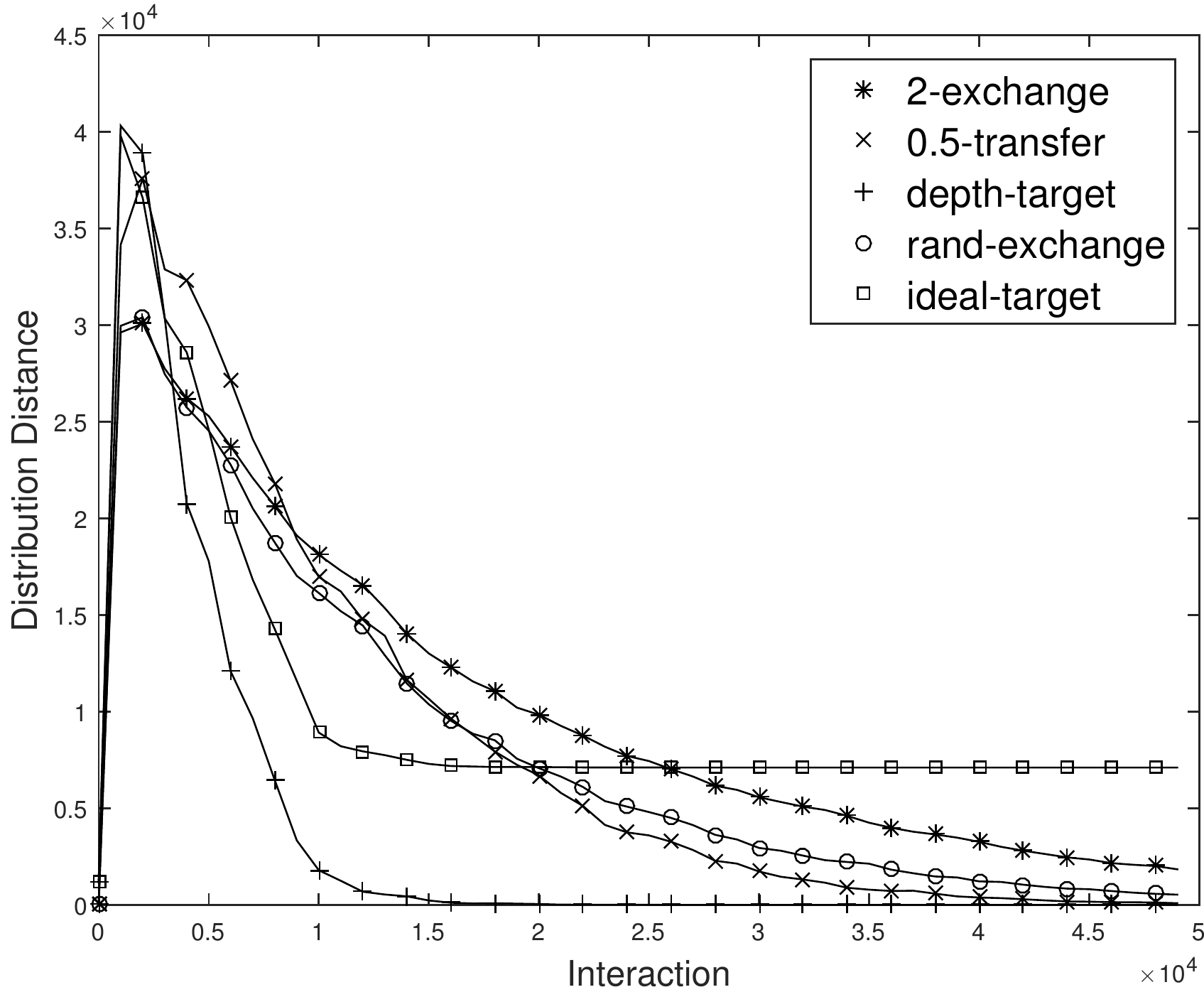}
   \caption{$n=50$}
   \label{fig:loss_non_uniform_50}
\end{subfigure}
\caption{Distribution distance of the five energy redistribution protocols ({\bf $2$-exchange}, {\bf $0.5$-transfer}, {\bf depth-target}, {\bf rand-exchange} and {\bf ideal-target}) in the {\em lossy} case for various number of nodes $n\in \{10,30,50\}$. Figures (a), (c) and (e) depict the performance of the protocols in the case where the nodes have the same initial energy, while Figures (b), (d) and (f) depict their performance in the case of different initial energy supplies.}
\label{fig:distribution_distance_lossy}
\end{figure*}

\begin{table*}[t]
\begin{center}
\begin{tabular}{|c|cccccccccccc|}
\noalign{\hrule height 1pt} \cline{1-13}
{\multirow{2}{*}{protocol} } & \multicolumn{6}{c|}{Same initial energy}   &  \multicolumn{6}{c|}{Different initial energy} \\ 
\cline{2-13}

& \multicolumn{2}{c|}{$n=10$} 
& \multicolumn{2}{c|}{$n=30$} 
& \multicolumn{2}{c|}{$n=50$} 
& \multicolumn{2}{c|}{$n=10$} 
& \multicolumn{2}{c|}{$n=30$} 
& \multicolumn{2}{c|}{$n=50$} \\
\noalign{\hrule height 1pt}\hline

{\bf $2$-exchange} & 
  \multicolumn{2}{c|}{6.93\%} & \multicolumn{2}{c|}{6.92\% } &  \multicolumn{2}{c|}{24.32\%} & \multicolumn{2}{c|}{6.15\%} & \multicolumn{2}{c|}{2.27\% } &  \multicolumn{2}{c|}{16.70\% }\\
\hline

{\bf $0.5$-transfer} & 
\multicolumn{2}{c|}{63.71\% } & \multicolumn{2}{c|}{66.32\% } &  \multicolumn{2}{c|}{56.73\% } & 
\multicolumn{2}{c|}{62.97\% } & \multicolumn{2}{c|}{66.94\%} &  \multicolumn{2}{c|}{64.02\% }\\
\hline      

{\bf depth-target} & 
\multicolumn{2}{c|}{33.81\%} & \multicolumn{2}{c|}{47.60\% } &  \multicolumn{2}{c|}{51.24\%} & 
\multicolumn{2}{c|}{33.14\% } & \multicolumn{2}{c|}{49.91\% } &  \multicolumn{2}{c|}{51.59\% }\\
\hline 

{\bf rand-exchange} & 
\multicolumn{2}{c|}{10.96\%} & \multicolumn{2}{c|}{11.97\%  } &  \multicolumn{2}{c|}{16.41\%} & 
\multicolumn{2}{c|}{11.21\%} & \multicolumn{2}{c|}{15.21\%} &  \multicolumn{2}{c|}{15.31\%}\\
\hline 

{\bf ideal-target} & 
\multicolumn{2}{c|}{0\%} & \multicolumn{2}{c|}{0\%} &  \multicolumn{2}{c|}{0\%} & 
\multicolumn{2}{c|}{0\%} & \multicolumn{2}{c|}{0\%} &  \multicolumn{2}{c|}{0\%}\\
      
\noalign{\hrule height 1pt}\hline
\end{tabular}  
\end{center}
\caption{Energy distance of the five energy redistribution protocols in the lossless case for various number of nodes.}
\label{tab:EnergyDistance-lossless}
\end{table*}  

\begin{table*}[t]
\begin{center}
\begin{tabular}{|c|cccccccccccc|}
\noalign{\hrule height 1pt} \cline{1-13}
{\multirow{2}{*}{protocol} } & \multicolumn{6}{c|}{Same initial energy}   &  \multicolumn{6}{c|}{Different initial energy} \\ 
\cline{2-13}

& \multicolumn{2}{c|}{$n=10$} 
& \multicolumn{2}{c|}{$n=30$} 
& \multicolumn{2}{c|}{$n=50$} 
& \multicolumn{2}{c|}{$n=10$} 
& \multicolumn{2}{c|}{$n=30$} 
& \multicolumn{2}{c|}{$n=50$} \\
\noalign{\hrule height 1pt}\hline

{\bf $2$-exchange} 
& \multicolumn{2}{c|}{18.73\% } 
& \multicolumn{2}{c|}{37.26\% } 
& \multicolumn{2}{c|}{38.11\%} 
& \multicolumn{2}{c|}{18.33\% } 
& \multicolumn{2}{c|}{33.56\% } 
& \multicolumn{2}{c|}{35.90\% } \\
\hline

{\bf $0.5$-transfer} 
& \multicolumn{2}{c|}{52.81\% } 
& \multicolumn{2}{c|}{57.27\% } 
& \multicolumn{2}{c|}{51.51\% } 
& \multicolumn{2}{c|}{53.83\% } 
& \multicolumn{2}{c|}{51.44\% } 
& \multicolumn{2}{c|}{49.71\% } \\
\hline      

{\bf depth-target} 
& \multicolumn{2}{c|}{38.43\% } 
& \multicolumn{2}{c|}{60.86\% } 
& \multicolumn{2}{c|}{65.47\% } 
& \multicolumn{2}{c|}{40.29\% } 
& \multicolumn{2}{c|}{59.69\% } 
& \multicolumn{2}{c|}{63.42\% } \\
\hline 

{\bf rand-exchange} 
& \multicolumn{2}{c|}{20.91\% } 
& \multicolumn{2}{c|}{39.57\% } 
& \multicolumn{2}{c|}{40.34\% } 
& \multicolumn{2}{c|}{21.11\% } 
& \multicolumn{2}{c|}{34.53\% } 
& \multicolumn{2}{c|}{37.19\% } \\
\hline 

{\bf ideal-target}
& \multicolumn{2}{c|}{20.53\%} 
& \multicolumn{2}{c|}{46.91\%} 
& \multicolumn{2}{c|}{49.59\%} 
& \multicolumn{2}{c|}{20.55\%} 
& \multicolumn{2}{c|}{41.44\%} 
& \multicolumn{2}{c|}{44.50\%} \\
      
\noalign{\hrule height 1pt}\hline
\end{tabular}  
\end{center}
\caption{Energy distance of the five energy redistribution protocols in the lossy case for various number of nodes.}
\label{tab:EnergyDistance-lossy}
\end{table*}  

\begin{table*}[t]
\begin{center}
\begin{tabular}{|c|cccccccccccc|}
\noalign{\hrule height 1pt} \cline{1-13}
{\multirow{2}{*}{protocol} } & \multicolumn{6}{c|}{Same initial energy}   &  \multicolumn{6}{c|}{Different initial energy} \\ 
\cline{2-13}

& \multicolumn{2}{c|}{$n=10$} 
& \multicolumn{2}{c|}{$n=30$} 
& \multicolumn{2}{c|}{$n=50$} 
& \multicolumn{2}{c|}{$n=10$} 
& \multicolumn{2}{c|}{$n=30$} 
& \multicolumn{2}{c|}{$n=50$} \\
\noalign{\hrule height 1pt}\hline

{\bf $2$-exchange} 
& \multicolumn{2}{c|}{34.30\% } 
& \multicolumn{2}{c|}{56.76\% } 
& \multicolumn{2}{c|}{65.08\%} 
& \multicolumn{2}{c|}{31.62\% } 
& \multicolumn{2}{c|}{64.01\% } 
& \multicolumn{2}{c|}{68.22\% } \\
\hline

{\bf $0.5$-transfer} 
& \multicolumn{2}{c|}{47.35\% } 
& \multicolumn{2}{c|}{66.59\% } 
& \multicolumn{2}{c|}{77.33\% } 
& \multicolumn{2}{c|}{43.86\% } 
& \multicolumn{2}{c|}{72.53\% } 
& \multicolumn{2}{c|}{79.93\% } \\
\hline      

{\bf depth-target} 
& \multicolumn{2}{c|}{17.66\% } 
& \multicolumn{2}{c|}{23.59\% } 
& \multicolumn{2}{c|}{29.01\% } 
& \multicolumn{2}{c|}{18.64\% } 
& \multicolumn{2}{c|}{30.04\% } 
& \multicolumn{2}{c|}{32.70\% } \\
\hline 

{\bf rand-exchange} 
& \multicolumn{2}{c|}{36.22\% } 
& \multicolumn{2}{c|}{58.93\% } 
& \multicolumn{2}{c|}{69.97\% } 
& \multicolumn{2}{c|}{33.48\% } 
& \multicolumn{2}{c|}{65.77\% } 
& \multicolumn{2}{c|}{72.76\% } \\
\hline 

{\bf ideal-target}
& \multicolumn{2}{c|}{12.16\%} 
& \multicolumn{2}{c|}{21.83\%} 
& \multicolumn{2}{c|}{27.16\%} 
& \multicolumn{2}{c|}{11.20\%} 
& \multicolumn{2}{c|}{24.53\%} 
& \multicolumn{2}{c|}{27.32\%} \\
      
\noalign{\hrule height 1pt}\hline
\end{tabular}  
\end{center}
\caption{Energy loss until convergence time for the five energy redistribution protocols in the lossy case.}
\label{tab:Energy-loss}
\end{table*}

\section{Conclusions and possible extensions}\label{sec:conclusion}
In this paper, we have thoroughly studied the problem of energy-aware tree network formation, both theoretically as well as via computer simulations. We proposed simple interaction protocols for the formation of arbitrary and $k$-ary tree networks, as well as energy redistribution protocols that exploit different levels of information regarding the network structure and achieve different energy distributions. Still, our work leaves open several problems and reveals new ones.

\paragraph{Energy distributions}
Aiming for an energy distribution that requires every node to have exactly or at least twice the energy of each of its children as we did in this work seems like an intuitive choice when we deal with tree networks and, of course, one can define several natural generalizations of the energy distributions we considered. For instance, given a parameter $\alpha \geq 1$, one can define the $\alpha$-exact energy distribution requiring that all nodes have exactly $\alpha$ times the energy of each of their children; similarly one can define parameterized versions of relaxed and the exact up to the root energy distributions. Then, the energy redistribution protocols presented in Section~\ref{sec:energy-protocols} can be adapted (using corresponding parameters) in order to aim for such generalized energy distributions. 
Moreover, one can take it a step further and consider different parameters for different nodes: node $i$ can be associated with a parameter $\alpha_i \geq 1$ and the requirement is that its energy is (at least) $\alpha_i$ times the energy of each of its children. This extension is probably the most realistic one and can be used to model data transmission scenarios with diverse energy consumption for the nodes that are closer to the root.

As we observed in our simulations in Section~\ref{sec:experiments}, the main problem with the definition of such multiplicative energy distributions is that almost all of the network energy gets concentrated close to the root of the tree, even for exact distributions (for example, consider the case of complete binary trees). Therefore, considering other interesting energy distributions, possibly with {\em additive} terms instead of multiplicative ones, is an important future direction. For instance, is it possible to achieve an energy distribution where every node has exactly or at least $\gamma$ units more energy than each of its children, for some parameter $\gamma >0$?
More generally, what happens if every node is associated with a different such additive parameter?

\paragraph{Network structures}
In the paper of \citet{madhja_mswim2016}, the star network structure is very well defined: one of the nodes is the center, while all others all peripherals. In contrast, our tree network formation protocols may construct so many different tree structures, ranging from totally balanced trees to lines or even stars. This is one of the reasons why the final energy distribution may end up be so different than the ideal one, especially when the number of nodes is large. Can we design protocols that are able to construct predictable tree networks? If so, then it may be possible to achieve energy distributions of better quality.  

Of course, it is also very interesting to consider totally different network structures and corresponding tailor-made energy distributions. For instance, can we design distributed protocols for the formation of grid networks and energy distributions that require more central nodes to have more energy than outer ones (thinking the grid as an approximate three-dimensional pyramid)? Since the nodes interact only in pairs, one of the several bottlenecks in forming grid networks is that the nodes cannot understand if they are at a corner of the structure or somewhere in the middle. Taking care of issues like these is non-trivial, and one may need to assume some sort of global information.   

\paragraph{Battery constraints}
Another issue that has not been addressed in the current paper or in previous related work, is that the nodes actually have battery limits. Therefore, it is not possible for a node (e.g. the root of a tree network) to store a really high amount of energy. This gives rise to many interesting algorithmic challenges. For instance, suppose that a parent-child pair of nodes interacts, and in order to satisfy the required energy equilibrium, the parent must receive energy from the child. However, what happens if the parent has reached its battery limit? The child then has to dispose parts of its energy instead by transmitting it to some other unrelated node that can carry it. But then, other interaction may take place and the state of the energy distribution over the network may be completely different, meaning that the child node possibly has outdated information what can lead to ``confusion.'' Resolving such problems seems to be highly non-trivial and definitely deserves more investigation in the future.

\paragraph{Some practical issues}
Finally let us discuss some practical issues that are not captured by our model and methods. Since the nodes are mobile, building a fixed network might be problematic. Many nodes that are considered important (like the root of the tree) might rarely enter the communication range of other nodes, or nodes that are considered to be less important might physically move very close to many other nodes. In such cases, the roles of the nodes are clearly not well chosen and the network should be rebuilt. Taking care of this issue in a distributive manner seems to be a very challenging task, especially for nodes with limited computational power and memory.

Furthermore, we have assumed perfect conditions during the communication between pairs of nodes. However, in practice, the total time during which two nodes are in close proximity and critically affects the performance of energy transmission (and the total amount of energy that can be exchanged), while communication data may be lost due to faults in the network. Even though our protocols can still be applied in such cases (some interactions might not take place or some energy exchanges may be incomplete), they might be highly inefficient, and coming up with new protocols that can adapt (or are tailor-made) to such faults is necessary.

\bibliographystyle{named}
\bibliography{bibliography}

\end{document}